\documentclass{llncs}
\usepackage{hyperref}
\usepackage{amsmath}
\usepackage{amssymb}
\usepackage{thm-restate}
\usepackage{array}
\usepackage{tikz}
\usepackage{url}
\usepackage{microtype}
\usepackage{stmaryrd}
\usepackage[linesnumbered,ruled,vlined]{algorithm2e}
\SetKwInput{KwInput}{Input}

\usepackage{subcaption}
\usepackage{symrob}

\usepackage{longtable}

\spnewtheorem{assumption}{Assumption}{\bfseries}{\itshape}

\usetikzlibrary{automata,calc}
\usetikzlibrary{arrows,shapes,decorations}
\usetikzlibrary{patterns}

\def\post{\text{\rm Post}}

\def\zReset{\textrm{\tt Reset}}

\def\Bits{\mathcal{B}}

\def\bool{\ensuremath{\mathbb{B}}}

\def\minterms{\mathsf{Minterms}}
\def\Paths{\mathsf{Paths}}

\title{Abstraction Refinement Algorithms\\ for Timed Automata\thanks{This work was funded by ANR project Ticktac (ANR-18-CE40-0015) and by ERC grant EQualIS (StG-308087).}}
\author{Victor Roussanaly, Ocan Sankur, Nicolas Markey}
\authorrunning{V. Roussanaly et al.}
\institute{Univ Rennes, Inria, CNRS, IRISA -- Rennes, France}
\begin{document}

\maketitle

\begin{abstract}
  We present abstraction-refinement algorithms for model
  checking safety properties of timed automata.
  The~abstraction domain we consider
  abstracts away zones by restricting the set of clock constraints
  that can be used to define them, while the refinement procedure
  computes the set of constraints that must be taken into
  consideration in the abstraction so as to exclude a given spurious
  counterexample. We~implement this idea in two~ways: an~enumerative
  algorithm where a lazy abstraction approach is adopted, meaning that
  possibly different abstract domains are assigned to each exploration
  node; and a~symbolic algorithm where the abstract transition system
  is encoded with Boolean formulas.
\end{abstract}

\section{Introduction}

Model checking~\cite{Pnu77,CE82,CGP00,BK08} is an automated technique
for verifying that the set of behaviors of a computer system satisfies
a given property. Model-checking algorithms explore finite-state
automata (representing the system under study) in order to decide if
the property holds; if~not, the algorithm returns an explanation.
These algorithms have been extended to verify real-time systems
modelled as timed automata~\cite{AD90,ACD93},
an extension of finite automata with clock variables
to measure and constrain the amount of time elapsed between occurrences
of transitions.
The state-space exploration can be done by representing clock constraints efficiently using convex polyhedra called~\emph{zones}~\cite{BM83,BY04}.
Algorithms based on this data structure have been implemented in
several tools such as Uppaal~\cite{BDL+06}, and have been applied in various
industrial cases.

The well-known issue in the applications of model checking is the
\emph{state-space explosion} problem: the~size of the state space
grows exponentially in the size of the description of the system.
There are several sources for this explosion: the system might be made
of the composition of several subsystems (such as a distributed
system), it might contain several discrete variables (such~as in a
piece of software), or it might contain a number of real-valued clocks
as in our case.

Numerous
attempts have been made to circumvent this problem. Abstraction is a
generic approach that consists in simplifying the model under study,
so~as to make it easier to verify~\cite{Cousot77}. \emph{Existential} abstraction may
only add extra behaviors, so that when a safety property holds in an
abstracted model, it~also holds in the original model; if~on the other
hand a safety property fails to hold, the model-checking algorithms
return a witness trace exhibiting the non-safe behaviour: this
either invalidates the property on the original model, if the
trace exists in that model, or gives information about how to automatically refine
the abstraction. This~approach, named CEGAR (counter-example guided
abstraction refinement)~\cite{CGJ+03}, was further developed and
used, for instance, in software verification (BLAST~\cite{HJMS03}, SLAM~\cite{BR01},~...).

The CEGAR approach has been adapted to timed automata, e.g. in
\cite{formats2007-DKL,HZHSG-tase10}, but the abstractions considered
there only consist in removing clocks and discrete variables, and
adding them back during refinement.  So for most well-designed models,
one ends up adding all clocks and variables which renders the method
useless.  Two~notable exceptions are \cite{HSW-cav13}, in which the
zone extrapolation operators are dynamically adapted during the
exploration, and \cite{TM-formats17}, in which zones are refined when
needed using interpolants. Both approaches define ``exact''
abstractions in the sense that they make sure that all traces
discovered in the abstract model are feasible in the concrete model at
any time.

In this work, we consider a more general setting and study
\emph{predicate abstractions} on clock variables. Just like in
software model checking, we~define abstract state spaces using these
predicates, where the values of the clocks and their relations are
approximately represented by these predicates.  New predicates are
generated if needed during the refinement step.  We~instantiate our
approach by two algorithms. The~first one is a zone-based enumerative
algorithm inspired by the \emph{lazy abstraction} in software model
checking~\cite{HJMS02}, where we assign a possibly different abstract
domain to each node in the exploration.  The~second algorithm is based
on binary decision diagrams~(BDD): by~exploiting the observation that a
small number of predicates was often sufficient to prove safety
properties, we~use an efficient BDD encoding of zones similar to one
introduced in early work~\cite{SB-cav03}.

\begin{figure}[b]
  \centering
  \begin{subfigure}[t]{0.45\textwidth}
    \centering
  \begin{tikzpicture}
    \begin{scope}[scale=.45]
      \fill[black!40] (1,1) -- (2,1) -- (2,2) -- (1,2)-- cycle;
      \draw[-latex',line width=.6pt] (0,0) -- (3.5,0);
      \draw[-latex',line width=.6pt] (0,0) -- (0,3.5);
      \node[left] at (0,3) {$y$};
      \node at (4,0) {$x$};
      \begin{scope}[line width=0.1pt]
        \draw[-] (1,0) -- (1,3.5);
        \draw[-] (2,0) -- (2,3.5);
        \draw[-] (3,0) -- (3,3.5);
        \draw[-] (0,1) -- (3.5,1);
        \draw[-] (0,2) -- (3.5,2);
        \draw[-] (0,3) -- (3.5,3);
        \draw[-] (0,0) -- (3.5,3.5);
        \draw[-] (1,0) -- (3.5,2.5);
        \draw[-] (2,0) -- (3.5,1.5);
        \draw[-] (0,1) -- (2.5,3.5);
        \draw[-] (0,2) -- (1.5,3.5);
      \end{scope}
      \begin{scope}[line width=1.5pt,color=red]
        \draw[-] (0,1) -- (2.5,3.5);
        \draw[-] (2,0) -- (3.5,1.5);
        \draw[-] (2,0) -- (2,3.5);
        \draw[-] (0,2) -- (3.5,2);
        \draw[-] (3,0) -- (3,3.5);
      \end{scope}
    \end{scope}
    \begin{scope}[scale=.45,shift={(6,0)}]
      \fill[black!40] (0,0) -- (0,1) -- (1,2) -- (2,2) -- (2,0)-- cycle;
      \draw[-latex',line width=.6pt] (0,0) -- (3.5,0);
      \draw[-latex',line width=.6pt] (0,0) -- (0,3.5);
      \node[left] at (0,3) {$y$};
      \node at (4,0) {$x$};
      \begin{scope}[line width=0.1pt]
        \draw[-] (1,0) -- (1,3.5);
        \draw[-] (2,0) -- (2,3.5);
        \draw[-] (3,0) -- (3,3.5);
        \draw[-] (0,1) -- (3.5,1);
        \draw[-] (0,2) -- (3.5,2);
        \draw[-] (0,3) -- (3.5,3);
        \draw[-] (0,0) -- (3.5,3.5);
        \draw[-] (1,0) -- (3.5,2.5);
        \draw[-] (2,0) -- (3.5,1.5);
        \draw[-] (0,1) -- (2.5,3.5);
        \draw[-] (0,2) -- (1.5,3.5);
      \end{scope}
      \begin{scope}[line width=1.5pt,color=red]
        \draw[-] (0,1) -- (2.5,3.5);
        \draw[-] (2,0) -- (3.5,1.5);
        \draw[-] (2,0) -- (2,3.5);
        \draw[-] (0,2) -- (3.5,2);
        \draw[-] (3,0) -- (3,3.5);
      \end{scope}
    \end{scope}
  \end{tikzpicture}
    \caption{Abstraction of zone $1\leq x,y\leq 2$}\label{fig-1a}
  \end{subfigure}
\begin{subfigure}[t]{0.52\textwidth}
  \centering
  \begin{tikzpicture}
    \begin{scope}[scale=.45]
      \fill[black!40] (1,0) -- (2,1) -- (3,1) -- (2,0)-- cycle;
      \draw[-latex',line width=.6pt] (0,0) -- (3.5,0);
      \draw[-latex',line width=.6pt] (0,0) -- (0,3.5);
      \node[left] at (0,3) {$y$};
      \node at (4,0) {$x$};
      \begin{scope}[line width=0.1pt]
        \draw[-] (1,0) -- (1,3.5);
        \draw[-] (2,0) -- (2,3.5);
        \draw[-] (3,0) -- (3,3.5);
        \draw[-] (0,1) -- (3.5,1);
        \draw[-] (0,2) -- (3.5,2);
        \draw[-] (0,3) -- (3.5,3);
        \draw[-] (0,0) -- (3.5,3.5);
        \draw[-] (1,0) -- (3.5,2.5);
        \draw[-] (2,0) -- (3.5,1.5);
        \draw[-] (0,1) -- (2.5,3.5);
        \draw[-] (0,2) -- (1.5,3.5);
      \end{scope}
      \begin{scope}[line width=1.5pt,color=red]
        \draw[-] (0,1) -- (2.5,3.5);
        \draw[-] (2,0) -- (3.5,1.5);
        \draw[-] (2,0) -- (2,3.5);
        \draw[-] (0,2) -- (3.5,2);
        \draw[-] (3,0) -- (3,3.5);
      \end{scope}
    \end{scope}

    \begin{scope}[scale=.45,shift={(5,0)}]
      \fill[black!40] (0,0) -- (0,1) -- (1,2) -- (3,2) -- (3,1) -- (2,0) -- cycle;
      \draw[-latex',line width=.6pt] (0,0) -- (3.5,0);
      \draw[-latex',line width=.6pt] (0,0) -- (0,3.5);
      \node[left] at (0,3) {$y$};
      \node at (4,0) {$x$};
      \begin{scope}[line width=0.1pt]
        \draw[-] (1,0) -- (1,3.5);
        \draw[-] (2,0) -- (2,3.5);
        \draw[-] (3,0) -- (3,3.5);
        \draw[-] (0,1) -- (3.5,1);
        \draw[-] (0,2) -- (3.5,2);
        \draw[-] (0,3) -- (3.5,3);
        \draw[-] (0,0) -- (3.5,3.5);
        \draw[-] (1,0) -- (3.5,2.5);
        \draw[-] (2,0) -- (3.5,1.5);
        \draw[-] (0,1) -- (2.5,3.5);
        \draw[-] (0,2) -- (1.5,3.5);
      \end{scope}
      \begin{scope}[line width=1.5pt,color=red]
        \draw[-] (0,1) -- (2.5,3.5);
        \draw[-] (2,0) -- (3.5,1.5);
        \draw[-] (2,0) -- (2,3.5);
        \draw[-] (0,2) -- (3.5,2);
        \draw[-] (3,0) -- (3,3.5);
      \end{scope}
    \end{scope}
  \end{tikzpicture}
  \caption{Abstraction of zone $y\leq 1 \mathrel\wedge 1\leq x-y\leq 2$}\label{fig-1b}
\end{subfigure}
\caption{The abstract domain is defined by the clock constraints shown
  in thick red lines.  In each example, the abstraction of the zone
  shown on the left (shaded area) is the larger zone on the right.}
\label{fig:abstraction-basic}
\end{figure}
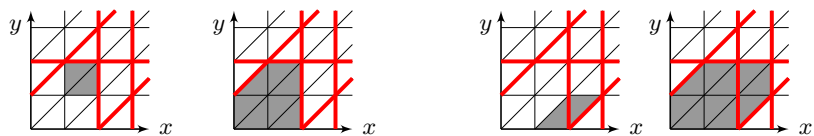

Let us explain the abstract domains we consider.
Assume there are two clock variables~$x$ and~$y$. The abstraction we
consider consists in restricting the clock constraints that can be
used when defining zones.  Assume that we only allow to compare~$x$
with~$2$ or~$3$; that $y$ can only be compared with~$2$, and $x-y$ can
only be compared with $-1$ or~$2$. Then any conjunction of constraints
one might obtain in this manner will be delimited by the thick red
lines in Fig.~\ref{fig:abstraction-basic}; one~cannot define a finer
region under this restriction.  The~figure shows the abstraction
process: given a ``concrete'' zone, its~abstraction is the smallest
zone which is a superset and is definable under our restriction.
For~instance, the~abstraction of $1\leq x,y\leq 2$ is $0\leq x,y\leq
2\land -1\leq x-y$ (cf.~Fig.~\ref{fig-1a}).

\paragraph{Related Works.}
We give more detail on zone abstractions in timed automata.  Most
efforts in the literature have been concentrated in designing zone
abstraction operators that are exact in the sense that they preserve
the reachability relation between the locations of a timed automaton;
see~\cite{BBLP-tacas04}. The~idea is to determine bounds on the
constants to which a given clock can be compared to in a given part of
the automaton, since the clock values do not matter outside these
bounds.  In~\cite{HKSW11,HSW-cav13}, the~authors give an algorithm
where these bounds are dynamically adapted during the exploration,
which allows one to obtain coarser abstractions.
In~\cite{TM-formats17}, the exploration tree contains pairs of zones:
a concrete zone as in the usual algorithm, and a coarser abstract
zone. The algorithm explores all branches using the coarser zone and
immediately refines the abstract zone whenever an edge which is
disabled in the concrete zone is enabled.  In~\cite{formats2010-EMP},
a CEGAR loop was used to solve timed games by analyzing strategies
computed for each abstract game. The abstraction consisted in
collapsing locations.

Some works have adapted the abstraction-refinement paradigm to timed
automata.  In~\cite{formats2007-DKL}, the authors apply ``localization
reduction'' to timed automata within an abstraction-refinement loop:
they abstract away clocks and discrete variables, and only introduce
them as they are needed to rule out spurious counterexamples.
A~more general but similar approach was developed in~\cite{HZHSG-tase10}.
In~\cite{WJ-atva14}, the authors adapt the trace abstraction refinement
idea to timed automata where a finite automaton is maintained to
rule out infeasible edge sequences. 

The CEGAR approach was also used recently in the LinAIG framework for
verifying linear hybrid automata~\cite{althaus-avocs-si17}. In this
work, the backward reachability algorithm exploits \emph{don't-cares}
to reduce the size of the Boolean circuits representing the state
space.  The~abstractions consist in enlarging the size of
\emph{don't-cares} to reduce the number of linear predicates used in
the representation.

\section{Timed Automata and Zones}
\label{section:definitions}
\subsection{Timed automata}
Given a finite set of clocks $\Clocks$, we~call \emph{valuations} the
elements of~$\Realnn^\Clocks$.
For a clock valuation~$v$, a~subset $R\subseteq \Clocks$, and a
non-negative real~$d$, we~denote with $v[R \leftarrow d]$ the
valuation~$w$ such that $w(x) = v(x)$ for $x \in \Clocks \setminus R$
and $w(x) = d$ for $x \in R$, and with $v+d$ the valuation~$w'$ such that
$w'(x)=v(x)+d$ for all~$x\in\Clocks$.
%
%
We~extend these operations to sets of valuations in the obvious
way. We~write~$\vec{0}$ for the valuation that assigns~$0$ to every
clock.
An~\emph{atomic guard} is a formula of the form $x \prec k$ or~$x-y
\prec k$ with $x,y \in \Clocks$, $k\in \mathbb{N}$, and~$\mathord\prec
\in \{\mathord{<},\mathord{\leq},\mathord{>},\mathord{\geq}\}$.
A~\emph{guard} is a conjunction of atomic guards.  A~valuation~$v$
satisfies a guard~$g$, denoted $v \models g$, if all atomic guards
hold true when each $x\in \Clocks$ is replaced
with~$v(x)$. Let~$\sem{g} = \{ v \in \Realnn^\Clocks \mid v \models
g\}$ denote the set of valuations satisfying~$g$.  We~write
$\Phi_\Clocks$ for the set of guards built on~$\Clocks$.

\smallskip
A~\emph{timed automaton} $\TA$ is a tuple
$(\Locs,\invar,\ell_0,\Clocks,E)$, where $\Locs$~is a finite set of
locations, $\invar\colon\Locs \rightarrow \Phi_\Clocks$ defines
location invariants, $\Clocks$~is a finite set of clocks, $E \subseteq
\mathcal{L} \times \Phi_\Clocks \times 2^\Clocks \times \mathcal{L}$
is a set of edges, and $\ell_0\in \Locs$ is the initial location.
An~edge $e = (\ell,g,R,\ell')$ is also written as $\ell
\xrightarrow{g, R} \ell'$.  For any location~$\ell$, we~let~$E(\ell)$
denote the set of edges leaving~$\ell$.

A~\emph{configuration} of~$\TA$ is a pair~$q=(\ell,v)\in \Locs\times
\Realnn^{\Clocks}$ such that $v\models\invar(\ell)$.
A~\emph{run} of~$\TA$ is a sequence $q_1e_1q_2e_2\ldots q_n$ where
for all~$i\geq 1$, $q_i=(\ell_i,v_i)$ is a configuration,
and either $e_i \in \mathbb{R}_{>0}$, in which case $q_{i+1} =
(\ell_i,v_i + e_i)$, or $e_i =(\ell_i,g_i,R_i,\ell_{i+1})\in E$, in
which case $v_i \models g_i$ and $q_{i+1} =
(\ell_{i+1},v_i[R_i\leftarrow 0])$.
%
A~\emph{path} is a sequence of edges with matching endpoint
locations.

\subsection{Zones and DBMs}
Several tools for timed automata implement algorithms based on \emph{zones},
which are particular polyhedra definable with clock constraints.
Formally, a zone~$Z$ is a subset of~$\Realnn^{\Clocks}$ definable by
a guard in~$\Phi_\Clocks$.

We recall a few basic operations defined on zones.  First, the
intersection~$Z\cap Z'$ of two zones~$Z$ and~$Z'$ is clearly a zone.
Given a zone~$Z$, the~set of time-successors of~$Z$, defined as
$\timesucc Z = \{v+t \in \Realnn^\Clocks \mid t\in\Realnn,\ v \in Z\}$, is
easily seen to be a~zone;
similarly for time-predecessors $\timepred Z = \{v \in \Realnn^\Clocks \mid
\exists t \geq 0.\ {v + t \in Z}\}$.  Given~$R\subseteq\Clocks$,
we~let $\reset_R(Z)$ be the zone $\{v[R\leftarrow 0] \in \Realnn^\Clocks
\mid v \in Z\}$, and
$\freeta_x(Z) = \{v' \in \Realnn^\Clocks \mid \exists v \in Z,
d \in \Realnn, v' = v[x \leftarrow d]\}$.

Zones can be represented as \emph{difference-bound
  matrices~(DBM)}~\cite{Dil89,BY04}.  Let~$\Clocksz = \Clocks \cup
\{0\}$, where~$0$ is an extra symbol representing a special clock
variable whose value is always~$0$.
A~DBM is a $\size\Clocksz \times \size\Clocksz$-matrix taking values in
$(\bbZ\times\{\mathord <,\mathord \leq\})\cup\{(+\infty,\mathord <)\}$.
Intuitively, cell~$(x,y)$ of a DBM~$M$ stores a pair~$(d,\prec)$
representing an upper bound on the difference~$x-y$. 
For any DBM~$M$, we~let~$\sem{M}$ denote the zone it defines.

While several DBMs can represent the same zone, each zone admits a
\emph{canonical} representation, which is obtained by storing the
tightest clock constraints defining the zone. This~canonical
representation can be obtained by computing shortest paths in a graph
where the vertices are clocks and the edges weighted by clock
constraints, with natural addition and comparison of elements
of~$(\bbZ\times\{\mathord <,\mathord \leq\})\cup\{(+\infty,\mathord
<)\}$. This graph has a negative cycle if, and only~if, the associated
DBM represents the empty zone.
%
%

All~the operations on zones can be performed efficiently
(in~$O(\size\Clocksz^3)$) on their associated DBMs while maintaining
reduced form.  For~instance, the~intersection $N=Z\cap Z'$ of two
canonical DBMs~$Z$ and~$Z'$ can be obtained by first computing the DBM
$M=\min(Z,Z')$ such that $M(x,y)=\min\{Z(x,y),Z'(x,y)\}$ for
all~$(x,y)\in\Clocksz^2$, and then turning~$M$ into canonical form.
We~refer to~\cite{BY04} for full details. By~a slight abuse of notation,
we~use the same notations for DBMs as for zones,  writing e.g.
$M'=\timesucc M$, where~$M$ and~$M'$ are reduced DBMs such that
$\sem{M'}=\timesucc{\sem{M}}$.
%
%
%
Given an edge~$e = (\ell,g,R,\ell')$, and a zone~$Z$, we define
$\postta_e(Z)=\invar(\ell') \cap
\timesucc{(g \cap \reset_R(Z))}$, and
$\preta_e(Z)=
\timepred{(g \cap \freeta_R(\invar(\ell') \cap Z))}$.
For a path~$\rho=e_1e_2\ldots
e_n$, we define $\postta_\rho$ and $\preta_\rho$ by iteratively
applying $\postta_{e_i}$ and $\preta_{e_i}$ respectively.

\subsection{Clock-predicate abstraction and interpolation}
For all clocks~$x$ and $y$
in~$\Clocksz$, we consider a finite set $\Dom_{x,y} \subseteq
\mathbb{N}\times\{\mathord\leq,\mathord<\}$, and gather these in a table $\Dom =
(\Dom_{x,y})_{x,y \in \Clocksz}$.  $\Dom$~is the \emph{abstract
  domain} which restricts zones to be defined only using constraints
of the form $x-y \prec k$ with~$(k,\mathord\prec) \in \Dom_{x,y}$, as seen
earlier.  Let~us call $\Dom$ the \emph{concrete domain} if $\Dom_{x,y}
= \mathbb{N}\times\{\mathord\leq,\mathord<\}$ for all~$x,y \in \Clocksz$.
A~zone~$Z$ is
$\Dom$-definable if there exists a DBM~$D$ such that $Z=\sem D$ and $D(x,y) \in
\Dom_{x,y}$ for all~$x,y\in\Clocksz$.  Note that we do not require
this witness DBM~$D$ to be reduced; the~reduction of such a DBM might
introduce additional values.
%
%
We~say that domain~$\Dom'$ is  a \emph{refinement} of~$\Dom$ if
for all~$x,y \in \Clocksz$, we~have $\Dom_{x,y}\subseteq \Dom'_{x,y}$.

An abstract domain~$\Dom$ induces an \emph{abstraction
function}~$\alpha_\Dom\colon {2^{\Realnn^\Clocks}}
\rightarrow {2^{\Realnn^\Clocks}} $ where $\alpha_\Dom(Z)$
is the smallest $\Dom$-definable zone containing~$Z$.  For~any
reduced~DBM~$D$, $\alpha_\Dom(\sem{D})$ can be computed by
setting~$D'(x,y) = \min\{(k,\prec) \in \Dom_{x,y} \mid D(x,y) \leq
(k,\prec)\}$ (with~$\min\emptyset=(\infty,<)$).



An \emph{interpolant} for a pair of zones~$(Z_1,Z_2)$ with~$Z_1\cap
Z_2 = \emptyset$ is a zone~$Z_3$ with~$Z_1 \subseteq Z_3$ and~$Z_3
\cap Z_2 = \emptyset$\footnote{It~is sometimes also required that
  the interpolant only involves clocks that have non-trivial constraints in both~$Z_1$ and~$Z_2$. We~do not impose this requirement in our definition, but it will hold true in the interpolants computed by our algorithm.}~\cite{TM17}.
%
We~use interpolants to refine our abstractions;
in~order not to add too many new constraints when refining, our~aim is
to find \emph{minimal interpolants}: define the density of a DBM~$D$
as $\dens D=\#\{(x,y) \in \Clocksz^2 \mid D(x,y)\not=(\infty,<)\}$.
Notice that while any pair of disjoint convex polyhedra can be
separated by hyperplanes, not all pairs of disjoint zones admit
interpolants of density~$1$; this is because not all (half-spaces
delimited~by) hyperplanes are zones.

\begin{restatable}{lemma}{lemmanosimpleinterp}
There exist pairs of zones accepting no simple interpolants.
\end{restatable}

\begin{proof}
  Consider 3-dimensional zones $A$, defined as $z=0\et x=y$, and $B$,
  defined as $y\geq 2 \et z\leq 2 \et y-x\leq 1 \et x-z\leq 1$. Both
  zones and their canonical DBMs are represented on
  Fig.~\ref{fig-nosimpleinterp}.
  \begin{figure}[h]
    \centering\noindent
    \begin{minipage}{.45\linewidth}
    \begin{tikzpicture}[scale=1.8,x={(0.45 cm,-0.05 cm)},z={(-0.3 cm,-0.15 cm)},y={(0,0.5 cm)}]
      \draw[-latex'] (0,0,0) -- (3,0,0) node[right] {$x$};
      \draw[-latex'] (0,0,0) -- (0,2.5,0) node[above] {$z$};
      \draw[-latex'] (0,0,0) -- (0,0,4) node[left] {$y$};
      \draw[line width=2pt,black!40!white] (0,0,0) -- (3.5,0,3.5) node[right,black] {$A$}; 
      \fill[black!30!white,draw=black!50!white,line width=2pt,opacity=.4] (1,0,2) -- (1,2,2) -- (3,2,2) -- cycle;
      \fill[black!30!white,draw=black!50!white,line width=2pt,opacity=.4] (1,0,2) -- (3,2,4) -- (3,2,2) node[above right,black,opacity=1] {$B$} -- cycle;
      \fill[black!30!white,draw=black!50!white,line width=2pt,opacity=.4] (1,0,2) -- (1,2,2) -- (3,2,4) -- cycle;
      \fill[black!30!white,draw=black!50!white,line width=2pt,opacity=.4] (1,2,2) -- (3,2,2) -- (3,2,4) -- cycle;
      \begin{scope}
        \draw[dashed] (1,0,2) -- (3,0,2) -- (3,0,4) -- cycle;
        \draw[dotted] (3,2,4) -- (3,0,4);
        \draw[dotted] (3,2,2) -- (3,0,2);
        \draw[dashed] (1,0,0) -- (3,2,0) -- (1,2,0) -- cycle;
        \draw[dotted] (1,0,2) -- (1,0,0);
        \draw[dotted] (3,2,2) -- (3,2,0);
        \draw[dotted] (1,2,2) -- (1,2,0);
        \draw[dashed] (0,0,2) -- (0,2,4) -- (0,2,2) -- cycle;
        \draw[dotted] (1,0,2) -- (0,0,2);
        \draw[dotted] (3,2,4) -- (0,2,4);
        \draw[dotted] (1,2,2) -- (0,2,2);
      \end{scope}
    \end{tikzpicture}
    \end{minipage}
    \begin{minipage}{.5\linewidth}
      \newcolumntype{C}{>{\centering $(}p{13mm}<{)$}}
      \[
      A = \left(\begin{array}{CCCC}
        0,\le & +\infty,< & +\infty,< & 0,\le \tabularnewline
        +\infty,< & 0,\le & 0,\le & +\infty,< \tabularnewline
        +\infty,< & 0,\le & 0,\le & +\infty,< \tabularnewline
        0,\le & +\infty,< & +\infty,< & 0,\le 
      \end{array}\right)
      \]
      \newcolumntype{C}{>{\centering $(}p{13mm}<{,\le)$}}
      \[
      B = \left(\begin{array}{CCCC}
        0 & 3 & 4 & 2 \tabularnewline
        -1 & 0 & 1 & 1 \tabularnewline
        -2 & 1 & 0 & 0 \tabularnewline
        0 & 1 & 2 & 0 
      \end{array}\right)
      \]
    \end{minipage}
    \caption{Two zones that cannot be separated by a simple interpolant}
    \label{fig-nosimpleinterp}
  \end{figure}
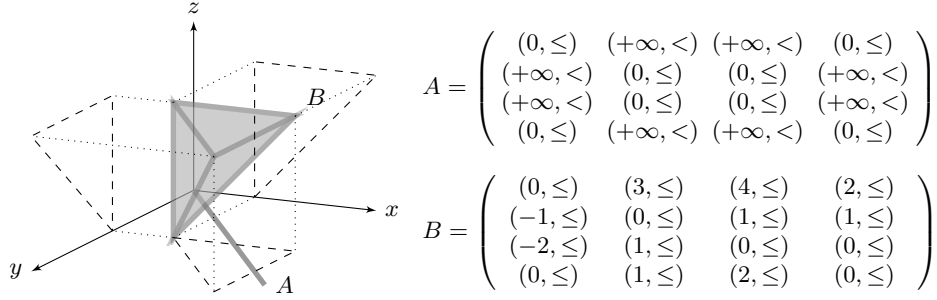

  We~observe that they are disjoint: if a triple~$(x,y,z)$ were in
  both~$A$ and~$B$, then $x=y$ and~$z=0$ (for being in~$B$); in~$A$,
  $y\geq2$, hence also $x\geq 2$, contradicting $x-z\leq 1$.

  Now, assume that there is a simple interpolant~$I$, with~$A\cap
  I=\emptyset$ and $B\subseteq I$. In~the canonical DBM of~$I$, only
  one non-diagonal element is not~$(+\infty,<)$; assume
  $I(x,y)\not=(+\infty,<)$. Then we~must have $A(y,x)+I(x,y) <
  (0,\mathord\leq)$, and $B(x,y)\leq I(x,y)$. Then $A(x,y)+B(x,y) <
  (0,\mathord\leq)$. However, it~can be observed that in our example,
  $A(x,y)+B(y,x)\geq (0,\mathord\leq)$ for all pairs~$(x,y)$.
\end{proof}

Still, we can bound the density of a minimal interpolant:
\begin{restatable}{lemma}{lemmasmallinterp}\label{lemma-smallinterpol}
For any pair of disjoint, non-empty zones~$(A,B)$, there exists an
interpolant of density less than or equal to~$\size{\Clocksz}/2$.
\end{restatable}

\begin{proof}
  Assume that $A$ and~$B$ are given as canonical DBMs, which we also
  write~$A$ and~$B$ for the sake of readability.  We prove the
  stronger result that $A\cap B=\emptyset$ if, and only~if, for
  some~$n\leq \size{\Clocksz}/2$, there exists a sequence of
  pairwise-distinct clocks $(x_i)_{0\leq i\leq 2n-1}$ such that,
  writing~$x_{2n}=x_0$, 
  \[
  \sum_{i=0}^{n-1} A(x_{2i},x_{2i+1}) + B(x_{2i+1},x_{2i+2}) < (0,\leq).
  \]
  Before proving this result, we~explain how we conclude the proof:
  the inequality above entails that
  \[
  \bigcap_{i=0}^{n-1} A(x_{2i},x_{2i+1})
  \cap
  \bigcap_{i=0}^{n-1} B(x_{2i+1},x_{2i+2}) = \emptyset
  \]
  where we abusively identify $A(x,y)$ with the half-space it
  represents.  It~follows that $\bigcap_{i=0}^{n-1}
  B(x_{2i+1},x_{2i+2})$ is an interpolant, whose density is less than
  or equal to~$\size\Clocksz/2$.
  
  Assume that such a sequence exists, and write~$C$ for the
  DBM $\min(A, B)$. Then
  \begin{xalignat*}1
    \sum_{i=0}^{2n-1} C(x_{i},x_{i+1}) &=
    \sum_{i=0}^{2n-1} \min \{A(x_{i},x_{i+1}),B(x_{i},x_{i+1})\}
    \\ &
    \le \sum_{i=0}^{n-1} A(x_{2i},x_{2i+1}) + B(x_{2i+1},x_{2i+2})
      < (0,\leq).
  \end{xalignat*}
  This entails that the intersection is empty.

  Conversely, if the intersection is empty, then there is a sequence
  of clocks $(x_i)_{0\leq i<m}$, with~$m\leq\size\Clocksz$, such that,
  letting $x_m=x_0$, we~have
  \[
  \sum_{i=0}^{m-1} \min\{A(x_i,x_{i+1}), B(x_i,x_{i+1})\} < (0,\leq).
  \]
  Consider one of the shortest such sequences.  Since~$A$ and~$B$ are
  non-empty, the sum must involve at least one element of each
  DBM. Moreover, if~it involves two consecutive elements of the same
  DBM (i.e., if $A(x_i,x_{i+1}) < B(x_i,x_{i+1})$ and
  $A(x_{i+1},x_{i+2}) < B(x_{i+1},x_{i+2})$ for some~$i$), then by
  canonicity of the DBMs of~$A$ and~$B$, we~can drop clock~$x_{i+1}$
  from the sequence and get a shorter sequence satisfying the same
  inequality, contradicting minimality of our sequence. The~result follows.
\end{proof}

By adapting the algorithm of~\cite{TM17} for computing interpolants,
we can compute minimal interpolants efficiently:
\begin{restatable}{proposition}{propniminterp}
Computing a minimal interpolant can be performed in~$O(\size{\Clocks}^4)$.
\end{restatable}

\begin{proof}
  \begin{figure}[t]
    \begin{algorithm}[H]
      \small
      \KwInput{canonical DBM $A$, $B$}\;
      \For{$(x,y)\in \Clocksz^2$}{%
        \eIf{$A(x,y)\leq B(x,y)$}
           {$M^0(x,y):=A(x,y)$\;}
           {$M^0(x,y):=(\infty,<)$\;}
      }
      $N^0:=\canonical(M^0)$\;
      \For{($i=1$; $i\leq \size\Clocksz/2$; $i++$)}
          {
            \For {$(x,y)\in \Clocksz^2$}
                 {
                   $M^i(x,y):= \min\{N^{i-1}(x,y),
                      \min_{z\in\SC_B(y)} N^{i-1}(x,z)+B(z,y)\}$\;
                 }
            \For {$(x,y)\in \Clocksz^2$}
                 {
                   $N^i(x,y):= \min\{M^{i}(x,y),
                      \min_{z\in\SC_A(y)} M^{i}(x,z)+A(z,y)\}$\;
                   \If{$(x=y)$ and $N^i(x,x)<(0,\mathord\leq)$}
                      {\Return ($\texttt{true}, i$)\;}
                 }
          }
          \Return \texttt{false}\;
      \caption{Algorithm for minimal interpolant}
      \label{alg:interpol}
      \label{alg-interpol}
    \end{algorithm}
\end{figure}

Algorithm~\ref{alg-interpol} describes our procedure. In~order to prove
its correctness, we~begin with proving that the sequence of DBM it computes
satisfies the following property:
\begin{lemma}
  For any~$i\geq 0$ such that $N^i$ has been computed by
  Algorithm~\ref{alg:interpol}, for any $(x,y)\in\Clocksz^2$, it~holds
  \[
  N^i(x,y) = \min_{\substack{\pi\in\Paths(x,y)\\\size\pi_B\leq i}} W_{\min(A,B)}(\pi).
  \]
\end{lemma}

\begin{proof}
The proof proceeds by induction on~$i$.  For~$i=0$, pick a
path~$\pi=(x_i)_{0\leq i\leq k}$ from~$x+0$ to~$x_k$ such that
$\size\pi_B=0$. Then
\begin{xalignat*}1
  W_{\min(A,B)}(\pi) &= \sum_{0\leq i<k} A(x_i,x_{i+1}) 
  = \sum_{0\leq i<k} M^0(x_i,x_{i+1}) \\
   &= \sum_{0\leq i<k} N^0(x_i,x_{i+1}) \geq N^0(x_0,x_k).
\end{xalignat*}
The first two equalities follow from the fact that $\pi$ only involves
transitions in~$E_{A\leq B}$; the third equality is because
canonization will not modify entries from~$A$ (since $A$ is originally
in canonical form). The~last inequality follows from canonicity of~$N^0$.

Now assume that the result holds at step~$i$, and that $N^{i+1}$ is
defined. Pick~$x$ and~$y$ in~$\Clocksz$. By~construction of~$M^{i+1}$
and~$N^{i+!}$, there exists $z$ and~$t$ in~$\Clocksz$ such that
$N^{i+1}(x,y) = N^i(x,z)+B(z,t)+A(t,y)$ with $t\in \SC_A(y)$
(or~$t=y$) and $z\in\SC_B(t)$ (or~$z=t$).  From the induction
hypothesis, there is a path~$\pi'$ from~$x$ to~$z$ such that
$N^i(x,z)=W_{\min(A,B)}(\pi')$, and~$\size{\pi'}_B\leq i$. Adding~$t$
and~$y$ to this path, we~get a path~$\pi$ from~$x$ to~$y$ such that
$N^{i+1}(x,y) = W_{\min(A, B)}(\pi)$ and $\size\pi_B\leq i+1$.

It~remains to prove that any path~$\pi$ from~$x$ to~$y$ with
$\size\pi_B\leq i+1$ is such that $N^{i+1}(x,y) \leq W_{\min(A,
  B)}(\pi)$. Fix such a path $\pi=(x_j)_{0\leq j\leq k}$\;
we~concentrate on the case where $\size\pi_B=i+1$, since the other
case follows from the induction hypothesis. We~decompose~$\pi$
as~$\pi_1=(x_j)_{0\leq j\leq l}$, $\pi_2=(x_l,x_{l+1})$ and
$\pi_3=(x_j)_{l+1\leq j\leq k}$, such that $(x_l,x_{l+1})\in E_B$ and
$(x_j,x_{j+1})\in E_A$ for all~$j>l$; in~other terms, $\pi_2$ is the
last $E_B$ transition of~$\pi$, and $\pi_3$~is a path from~$x_{l+1}$
to~$x_k$ only involving transitions in~$A$. Then
\[
w_{\min(A, B)}(\pi_2\cdot\pi_3) = B(x_l,x_{l+1}) + w_A(\pi_3)
\geq  B(x_l,x_{l+1}) + w_{A}(x_{l+1},x_k),
\]
and
\begin{itemize}
\item if $(x_{l+1},x_k)\in E_B$, then $w_{\min(A,B)}(\pi_2\cdot\pi_3)
  \geq w_B(x_l,x_k)$, and, applying the induction hypothesis, 
  $w_{\min(A, B)}(\pi) \geq N^i(x,x_l)+\min\{A(x_l,x_k), B(x_l,x_k)\}$.
  Since $M^{i+1}(x,x_l)\leq N^i(x,x_l)$, we~get
  \begin{xalignat*}1
  w_{\min(A, B)}(\pi) &\geq \min\{N^i(x,x_l)+B(x_l,x_k), M^{i+1}(x,x_l)+A(x_l,x_k)\}
  \\ &\geq N^{i+1}(x,x_k).
  \end{xalignat*}
\item if $(x_{l+1},x_k)\in E_B$, then $w_{\min(A,B)}(\pi_2\cdot \pi_3)
  \geq B(x_l,x_{l+1}) + A(x_{l+1}, x_k)$. From the induction
  hypothesis, $w_{\min(A, B)}(\pi) \geq N^i(x,x_l)+B(x_l,x_{l+1}) +
  A(x_{l+1},x_k) \geq N^{i+1}(x,x_k)$.\qed
\end{itemize}
\let\qed\relax
\end{proof}

Following the argument of the proof of Lemma~\ref{lemma-smallinterpol}, we~get:
\begin{corollary}
If $A \cap B\not=\emptyset$, then Algorithm~\ref{alg:interpol} returns
\texttt{\textup{false}}; otherwise, it~returns
$(\texttt{\textup{true}},k)$ for the smallest~$k$ such that for all
cyclic path~$\pi$ such that $\size\pi_B<k$, it~holds $w_{\min(A, B)}(\pi)\geq (0,\leq)$.
\end{corollary}
This entails that $k$ is the dimension of the minimal interpolant. The
minimal interpolant can be obtained by taking the $B$-elements of the
negative cycle found by the algorithm.
\end{proof}

\section{Enumerative Algorithm}
\label{sec-enum}
The first type of algorithm we present is a zone-based enumerative
algorithm based on the clock-predicate abstractions.  Let us first
describe the overall algorithm in Algorithm~\ref{alg:cegar}, which is a
typical abstraction-refinement loop.  We~then explain how the
abstract reachability and refinement procedures are instantiated.


\begin{figure}[ht]
  \begin{minipage}[t]{0.47\textwidth}
    \begin{algorithm}[H]
      \small
      \KwInput{$\TA = (\Locs, \invar, \ell_0, \Clocks, E)$, 
        $\ell_T$}
      Initialize~$\Dom_0$\;
      \label{alg:cegar:init}
      \wait := $\{\node(\ell_0,\vec{0}{\uparrow}, \Dom_0)\}$\;
      \label{alg:cegar:wait-init}
      \passed := $\emptyset$\;
      \While{}
      {%
        $\pi :=\AbsReach(\TA, \wait,\penalty0 \passed, \ell_T)$\;
        \eIf{$\pi = \emptyset$}{\Return Not reachable\;}%
        {%
          \eIf{trace~$\pi$ is feasible}{\label{alg:cegar:feas}\Return Reachable\;}%
          {\hbox to 10pt{$\Refine(\pi, \wait, \passed)$\;\hss}}%
        }%
      }%
      \Return Not reachable\;
      \caption{Enumerative algorithm checking the reachability of a target location~$\ell_T$.}
      \label{alg:cegar}
    \end{algorithm}
  \end{minipage}
  \hfill
  \begin{minipage}[t]{0.55\textwidth}
    \begin{algorithm}[H]
      \small
      \KwInput{$(\Locs, \invar, l_0, \Clocks, E)$, \wait, \passed, $\ell_T$}
      \While{$\wait \neq \emptyset$}
      {
        $n := \wait.pop()$\;
        \If{$n.\ell = \ell_T$}{
          \Return Trace from root to~$n$\;
        }
        \eIf{$\exists n' \in \passed \text{ such that } n.\ell = n'.\ell \land n.Z \subseteq n'.Z$}{
          $n.\covered := n'$\;
        }{
          $n.Z := \alpha(n.Z, n)$\;
          \label{alg:line:alpha}
          $\passed.add(n)$\;
          \For{$e=(\ell,g,R,\ell') \in E(n.\ell)$ s.t.
            $Z':=\post_e(n.Z) \neq \emptyset$}{
            $\Dom':=\choosedom(n,e)$\;            \label{alg:line:choose-dom}
            $n' := \node(\ell',Z',\Dom')$\;
            $n'.\parent := n$\;
            $\wait.add(n')$\;
          }
        }
      }
      \Return $\emptyset$\;
      \caption{\AbsReach
      }
      \label{alg:abs-reach}
    \end{algorithm}

  \end{minipage}
\end{figure}

The initialization at line~\ref{alg:cegar:init} chooses an abstract
domain for the initial state, which can be either empty (thus the
coarsest abstraction) or defined according to some heuristics.
The algorithm maintains the \wait and \passed lists that are used in
the forward exploration.  As~usual, the~\wait list can be implemented
as a stack, a queue, or another priority list that determines the
search order.
The algorithm also uses covering nodes. Indeed if there are two node
$n$ and $n'$, with $n\in \passed$, $n'\in \wait$, $n.\ell=n'.\ell$,
and $n'.z\subseteq n.Z$, then we know that every location reachable
from $n'$ is also reachable from $n$. Since we have already explored
$n$ and we generated its successors, there is no need to explore the
successors of $n'$.  The algorithm explicitly creates an exploration
tree: line~\ref{alg:cegar:wait-init} creates a node containing
location~$\ell_0$, zone~$\vec{0}{\uparrow}$, and the abstract
domain~$\Dom_0$ as the root of our tree, and adds this to the \wait
list.  More details on the tree are given in the next subsection.
Procedure~\AbsReach then looks for a trace to the target
location~$\ell_T$.
If~such a trace exists, line~\ref{alg:cegar:feas} checks its
feasibility.
Here $\pi$ is a sequence of node and edges of~$\TA$.  The feasibility
check is done by computing predecessors with zones starting from the
final state, without using the abstraction function.  If the last zone
intersects our initial zone, this means that the trace is
feasible. More details are given in Section~\ref{sec:refine}.

\subsection{Abstract forward reachability: \AbsReach}
We give a generic algorithm independently from the implementations of
the abstraction functions and the refinement procedure.

Algorithm~\ref{alg:abs-reach} describes the reachability procedure
under a given abstract domain~$\Dom$.  It~is similar to the standard
forward reachability algorithm using a \wait-list and a \passed-list.
We~explicitly create an exploration tree where the leaves are nodes in $\wait$, covered nodes, or nodes that have no non-empty successors. Each node~$n$ contains the
fields $\ell,Z$ which are labels describing the current location and
zone; field \covered points to a node covering the current node
(it~is undefined if the current node is not (known to~be) covered);
field \parent
points to the parent node in the tree (it~is undefined for the root);
and field~$\Dom$ is the abstract domain associated with the node.
Thus, the algorithm uses a possibly different abstract domain for
each node in the exploration tree.

The difference of our algorithm w.r.t. the standard reachability can
be seen at lines~\ref{alg:line:alpha} and~\ref{alg:line:choose-dom}.
At~line~\ref{alg:line:alpha}, we~apply the abstraction function to the
zone taken from the \wait-list before adding it to the \passed-list.
The abstraction function~$\alpha$ is a function of a zone~$Z$ and a
node~$n$.  This allows one to define variants with different
dependencies; for instance, $\alpha$ might depend on the abstract
domain $n.\Dom$ at the current node, but it can also use other
information available in~$n$ or on the path ending in~$n$.  For~now,
it~is best to think of~$\alpha$ simply as~$Z \mapsto
\alpha_{n.\Dom}(Z)$.
At line~\ref{alg:line:choose-dom}, the function \choosedom chooses an
abstract domain for the node~$n'$. 

%
%
%
%
In our implementation, the abstraction function always abstracts the
given zone w.r.t. the abstract domain~$n.\Dom$. For~\choosedom, we
considered three variants:
\begin{itemize}
\item one using a global domain, the same for all nodes: this way,
  each refinement benefits to all nodes, but this is often a drawback
  since in general different parts of the automaton will better have
  different abstract domain;
\item one using a local domain for each node: this has the
  advantage of using the coarsest possible abstraction, but it takes
  more memory and usually involves more refinements.
\item one using one domain per location of the automaton: this appears
  to be a good trade-off between the above two approaches.
\end{itemize}

\begin{remark}
Note that we use the abstraction function when the node is inserted in
the \passed list. This is because we want the node to
contain the smallest zone possible when we test whether the node is
covered. We~only need to use the abstracted zone when we
compute its successor and when we test whether the node is
covering. This allows us to store a unique zone.
\end{remark}


As a first step towards proving correctness of our algorithm, we consider that
the following property is preserved by Algorithm~\AbsReach:
\begin{equation}
  \begin{minipage}{.86\linewidth}
    For all nodes $n$ in \passed, for all edges $e$
    from~$n.\ell$, if $\postta_{e}(n.Z)\neq \emptyset$, then $n$ has a
    child~$n'$ such that $\postta_{e}(n.Z)\subseteq n'.Z$. If~$n'$ is
    in \passed, then we also have
    $\alpha_{n'.\Dom}(\postta_{e}(n.Z))\subseteq n'.Z$.
  \end{minipage}\hskip5mm
  \label{Prop1}
\end{equation}

The following is an easy observation about our algorithm:
\begin{lemma}\label{lemma-absr}
  Algorithm~\AbsReach preserves Property~\eqref{Prop1}.
\end{lemma}

Note that although we use inclusion in Property~\eqref{Prop1},
\AbsReach would actually preserve equality of zones, but we will not
always have equality before running \AbsReach. This is because \Refine
might change the zones of some nodes without updating the zones of all
their descendants.


\subsection{Refinement: \Refine}\label{sec:refine}
\begin{figure}[t]
  \begin{minipage}[t]{0.47\textwidth}
    \begin{algorithm}[H]
      \KwInput{$\pi$, 
        \wait, \passed}\;
        $n:= \text{last node of $\pi$}$\;
        $Z:=n.Z$\;
        $r:= \Refinerec(n,Z,\wait,\passed)$\; 
        $n_{cut}:=$ node to cut \mbox{(according to heuristics)}\;
        $\cut(n_{cut})$\;
        \uIf{$n_{cut}.Z=\emptyset$} {
        	\textit{delete} $n_{cut}$
        }
        \Else{
        	$\passed.remove(n_{cut})$\;
        	$\wait.add(n_{cut})$\;
        	$n_{cut}.Z:=\Concrete(n_{cut})$\;
        }
        \Return $r$\;
        
        \caption{\Refine}
    \end{algorithm}

  \end{minipage}
  \hfill
  \begin{minipage}[t]{0.52\textwidth}
    \begin{algorithm}[H]
      \small
      \KwInput{$n,Z,\wait,\passed$}\;
      $C:=\Concrete(n)$\;
      \uIf{$C\cap Z = \emptyset$} {

      	\Strengthen($n,Z,C,\wait$)\;\label{refine:modif1}
		\Return Not Feasible\;
	  }
      \uElseIf{$n$ has no parent}{
       		\Return Feasible \;
      	}     	
      \Else{
      		$e$ edge from $n.\parent$ to $n$\;
      		$Z':=\preta_e(Z)\cap n.\parent.Z$\;
      		\uIf{$\Refinerec(n.\parent,Z',\wait,\penalty-100\relax\passed)$= \textup{Feasible}}{
      			\Return Feasible\;
      		}\Else {
      			$C:=$\textsf{Concrete}$(n)$\;
      			\textsf{Strengthen}($n,Z,C,\wait$)\;\label{refine:modif2}
      			\Return Not Feasible\;
      		}
      }
      
      \caption{\Refinerec}
    \end{algorithm}

  \end{minipage}

    \begin{minipage}[t]{0.47\textwidth}
        \begin{algorithm}[H]
      \KwInput{$n$}\;
        \uIf{$n$ has parent} {
        	$e:=$ edge from $n.\parent$ to $n$\;
        	\Return $\postta_e(n.\parent.Z)$;
        }
        \Else{
        	\Return initial zone\;
        }

      \caption{\Concrete}
    \end{algorithm}
    \end{minipage}\hfill
    \begin{minipage}[t]{0.52\textwidth}

    \begin{algorithm}[H]
      \KwInput{$n,Z,C,\wait$}\;
      	\If{$\alpha_{n.\Dom}(C)\cap Z\neq \emptyset$}{

			$I:=$\textsf{interpolant}$(C,Z)$\;
			$n.\Dom.add(I)$\;
			
		}
		$n:=\alpha_{n.\Dom}(C)$\;
		Add every uncovered nodes to \wait\;
		      
      \caption{\Strengthen}
    \end{algorithm}
		
  \end{minipage}
\end{figure}

We now describe our refinement proecdure \Refine.
Let us now assume that \AbsReach returns $\pi = A_1
\xrightarrow[]{\sigma_1} A_2 \xrightarrow{\sigma_2} \ldots
\xrightarrow{\sigma_{k-1}} A_k$, and write $\Dom_i$ for the domain
associated with each~$A_i$. We~write $C_1$ for the initial concrete zone, and
for $i< k$, we~define $C_{i+1}=\postta_{\sigma_i}(A_i)$. We~also note $Z_k=A_k$
and for $i<k$, $Z_i= \preta_{\sigma_i}(Z_{i+1}) \cap A_i$. Then
$\pi$ is not feasible if, and only~if,
$\postta_{\sigma_1\ldots\sigma_k}(C_1)=\emptyset$, or equivalently
$\preta_{\sigma_1\ldots\sigma_k}(A_k)\cap C_1=\emptyset$. Since
for all $i<k$, it~holds $C_i\subseteq A_{i+1}$, we~have that
$\pi$ is not feasible if, and only~if,
$\exists i\leq k.\ C_i\cap Z_i=\emptyset$.
We~illustrate this on Fig.~\ref{fig-spurious}.

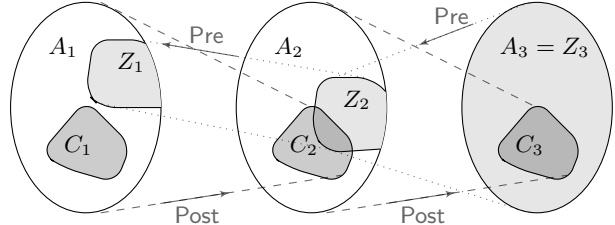
\begin{figure}[ht]
  \centering
  \begin{tikzpicture}
    \path[use as bounding box] (0,-.6) -- (7.6,2);
    \begin{scope}
    \draw[clip] (.8,.8) ellipse (1cm and 1.4cm);
    \draw[rounded corners=2mm,fill=black!10!white] (2.5,0.8) --
      (1,0.8) -- (.8,1) --
      (.9,1.7) -- (1.6,1.7) -- cycle;
    \draw[rounded corners=2mm,fill=black,opacity=.2] (.5,0) --
    (1.2,-.2) node[coordinate,above=1mm] (A1) {} -- (1.4,.3) --
    (.8,.9) node[coordinate,below=1mm] (B1) {} --
    (.2,.2) -- cycle;
    \draw[rounded corners=2mm] (.5,0) --
    (1.2,-.2) node[coordinate,above=1mm] (A1) {} -- (1.4,.3) --
    (.8,.9) node[coordinate,below=1mm] (B1) {} --
    (.2,.2) -- cycle;
    \path (1.4,1.4) node {$Z_1$};
    \path (.7,.3) node {$C_1$};
    \path (.5,1.6) node {$A_1$};
    \end{scope}
    \begin{scope}[xshift=3cm]
      \draw[clip] (.8,.8) ellipse (1cm and 1.4cm);
      \draw[rounded corners=2mm,fill=black!10!white] (2.5,0.3) --
      (1,0.2) -- (.8,.5) --
      (.9,1.2) -- (1.6,1.2) -- cycle;
      \draw[rounded corners=2mm,fill=black,opacity=.2] (.5,0) --
      (1.2,-.2)  node[coordinate,above=1mm] (A2) {} -- (1.4,.3) --
      (.8,.9) node[coordinate,below=1mm] (B2) {} --
      (.2,.2) -- cycle;
      \draw[rounded corners=2mm] (.5,0) --
      (1.2,-.2)  node[coordinate,above=1mm] (A2) {} -- (1.4,.3) --
      (.8,.9) node[coordinate,below=1mm] (B2) {} --
      (.2,.2) -- cycle;
      \path (1.4,.9) node {$Z_2$};
      \path (.7,.3) node {$C_2$};
      \path (.5,1.6) node {$A_2$};
    \end{scope}
    \begin{scope}[xshift=6cm]
    \draw[fill=black!10!white] (.8,.8) ellipse (1cm and 1.4cm);
      \draw[rounded corners=2mm,fill=black,opacity=.2] (.5,0) --
      (1.2,-.2)   node[coordinate,above=1mm] (A3) {} -- (1.4,.3) --
      (.8,.9) node[coordinate,below=1mm] (B3) {} --
      (.2,.2) -- cycle; 
      \draw[rounded corners=2mm] (.5,0) --
      (1.2,-.2)   node[coordinate,above=1mm] (A3) {} -- (1.4,.3) --
      (.8,.9) node[coordinate,below=1mm] (B3) {} --
      (.2,.2) -- cycle;
      \path (.7,.3) node {$C_3$};
    \path (.5,1.6) node {$A_3\hbox to 0pt{${}=Z_3$}$};
    \end{scope}

    \draw[opacity=.6,dashed] (1,2.2) -- (B2);
    \draw[opacity=.6,dashed] (1,-.6) -- (A2);
    \draw[opacity=.6,shorten >=15mm,shorten <=8mm,-latex'] (1,-.6) -- (A2)
      node[below,pos=.4] {$\postta$};

    \draw[opacity=.6,dashed] (4,2.2) -- (B3);
    \draw[opacity=.6,dashed] (4,-.6) -- (A3);
    \draw[opacity=.6,shorten >=15mm,shorten <=8mm,-latex'] (4,-.6) -- (A3)
      node[below,pos=.4] {$\postta$};

    \draw[opacity=.6,dotted] (6.7,2.2) -- (4.1,1.2);
    \draw[opacity=.6,dotted] (6.7,-.6) -- (4,0.25);
    \draw[opacity=.6,shorten >=12mm,shorten <=8mm,-latex']
      (6.7,2.2) -- (4.1,1.2) node[above,pos=.4] {$\preta$};

    \draw[opacity=.6,dotted] (4.5,1.2) -- (1.5,1.7);
    \draw[opacity=.6,dotted] (4,0.25) -- (1,0.85);
    \draw[opacity=.6,shorten >=3mm,shorten <=17mm,-latex']
      (4.5,1.2) -- (1.5,1.7) node[above,pos=.7] {$\preta$};

  \end{tikzpicture}
  \caption{Spurious counter-example: $Z_1 \cap C_1 = \emptyset$}
  \label{fig-spurious}
\end{figure}


Let us assume that $\pi$ is not feasible. Let us denote by $i_0$ the maximal
index such that $C_{i_0}\cap Z_{i_0}=\emptyset$. This index also has
the property that for all $j < i_0$, we~have $Z_j=\emptyset$ and $Z_{i_0}\neq
\emptyset$. Once we have identified this trace as spurious by
computing the $Z_j$, we have two possibilities:

\begin{itemize}
\item if $Z_{i_0} \cap \alpha_{\Dom_{i_0}}(C_{i_0})\neq \emptyset$: this
  means that we can reach $A_k$ from $\alpha_{\Dom_{i_0}}(C_{i_0})$
  but not from $C_{i_0}$. In other words, our abstraction is too
  coarse and we must add some values to $\Dom_{i_0}$ so that $Z_{i_0}
  \cap \alpha_{\Dom_{i_0}}(C_{i_0})= \emptyset$. Those values are
  found by computing the interpolant of $Z_{i_0}$ and $C_{i_0}$

\item Otherwise it means that $\alpha_{\Dom_{i_0}}(C_{i_0})$ cannot
  reach $A_k$ and the only reason the trace exists is because either
  $\Dom_{i_0}$ or $A_{i_0 -1}$ has been modified at some point and
  $A_{i_0}$ was not modified accordingly.
\end{itemize}

We can then update the values of $C_i$ for $i>i_0$ and repeat the
process until we reach an index~$j_0$ such that $C_{j_0}=
\emptyset$. We~then have modified the nodes $n_{i_0},\ldots ,n_{j_0}$ and
knowing that $n_{j_0}.Z=\emptyset$, we~can delete it and all of its
descendants. Since some of the descendants of $n_{i_0}$ have not been
modified, this might cause some refinements of the first type in the
future. In~order to ensure termination, we~sometimes have to cut a
subtree from a node in $n_{i_0},\ldots ,n_{j_0-1}$ and reinsert it in
the \wait list to restart the exploration from there. We~call this
action~\cut, and we can use several heuristics to decide when to
use~it. In~the rest of this paper we will use the following
heuristics:
we~perform~\cut on the first node
of $n_{i_0}...n_{j_0}$ that is
covered by some other node.
Since this node is covered, we know that we will not restart the
exploration from this node, or that the node was covered by one of its
descendant. If~none of these nodes are covered, we~delete~$n_{j_0}$ and
its descendants.
Other heuristics are possible, for instance applying \cut
on~$n_{i_0}$. We found that the above heuristics was the most efficient
in our experiments.

\begin{restatable}{lemma}{lemmainclusion}
\label{lemma:inclusion}
  Pick a node~$n$, and let $Y=n.Z$. Then after running~\Refine,
  either node~$n$ is deleted, or it holds  $n.Z\subseteq Y$.
  In~other words, the~zone of a node can only be reduced by \Refine.
\end{restatable}

\begin{proof}
\Refine may only add values to the domain of a
node, so that the refined zone is included in the previous one.
%
If~no values were added to the domain of a node, then its parent must
have been modified.  Since $A\subseteq B \Rightarrow
\postta_e(A)\subseteq \postta_e(B)$, the result follows
by induction.
\end{proof}

It follows that \Refine also preserves Property~\eqref{Prop1}, so that:
\begin{restatable}{lemma}{lemmacomplete}
  \label{lemma:complete}
  Algorithm~\ref{alg:cegar} satisfies Property~\eqref{Prop1}.
\end{restatable}

\begin{proof}
  We~prove that procedure \Refine preserves
  property~\eqref{Prop1}. Combined with Lemma~\ref{lemma-absr}, this
  entails the result.
  First notice that \Refine may not add new nodes. 
  Let~$n'$ be a node, $n$ its parent, and $e$ the edge from~$n$ to~$n'$.
  Three cases may arise:
  \begin{itemize}
  \item $n'$ has been modified: then it must have been modified at
    line \ref{refine:modif1} or \ref{refine:modif2}. At this point $n$
    is no longer modified, and \Refine ensures that
    $\Concrete(n)\subseteq n'.Z$, and
    $\Concrete(n)=\postta_{e}(n.Z)$. If $n'$ is in \passed, we also
    have $\alpha_{n'.\Dom}(\postta_{e}(n.Z))= n'.Z$;
  \item $n'$ has been deleted: in~this case, if $n$ is part of the
    subtree and it has either been deleted or moved to \wait, and is
    not in \passed anymore. Otherwise $n$ is not part of the subtree
    that has been cut and $n'$ is the root of this subtree, with
    $n'.Z=\emptyset$ and since $ \Concrete(n)\subseteq n'.Z$;
  \item $n'$ has not been modified, but $n$~has: then
    using Lemma~\ref{lemma:inclusion}, we~know that the inclusion
    $\postta_{e}(n.Z)\subseteq n'.Z$ (or
    $\alpha_{n'.\Dom}(\postta_{e}(n.Z))\subseteq n'.Z$) is preserved
    by \textsf{Refine}.\qed
\end{itemize}
\let\qed\relax
\end{proof}

We can then prove that our algorithm correctly decides the reachability problem
and always terminates.
\begin{restatable}{theorem}{thmalgoenum}
  Algorithm~\ref{alg:cegar} terminates and
  is correct.
\end{restatable}

\begin{proof}
We~first prove correctness, assuming termination.  First let us notice
that if the \wait set is empty, then for any reachable location~$l$,
there is a node $n$ such that $n.\ell = l$. This is because
we over-approximate the zones as shown in Lemma~\ref{lemma:complete},
so we over-approximate the set of reachable states.
%
Thus, if $\AbsReach(\TA$, \wait, \passed, $\ell_T)$ returns
$\emptyset$, then $\ell_T$ is not reachable in~$\TA$. In~other words,
if the enumerative algorithm returns "\emph{Not reachable}" then
$\ell_T$ is indeed not reachable.

On the other hand, if the algorithm returns "\emph{Reachable}", it
means that there is a feasible trace reaching~$\ell_T$, and
$\ell_T$~is indeed reachable.

\medskip

%
We~now prove termination.
Since there are a finite number of possible locations and we can limit
the number of possible zones to a finite number using abstraction
functions, we can deduce that \AbsReach terminates.

Let us assume that the enumerative algorithm does not terminate. Then
it means that \Refine is called infinitely many timed. Note that
\Refine is modifying a node and a node can be modified a
finite number of~time.

We can also note that a node can be destroyed only if one of its
ancestors is modified. As~such, we~can show that for every depth~$k$,
there is a point in the algorithm where every node at depth~$k$ or
less is fixed and will no longer be modified.

So we know that the algorithm does not terminate if, and only~if,
the depth of the
resulting tree is unbounded. This means that there exists a path where
we have two distinct nodes $n_1$ and~$n_2$ with $n_1.\ell = n_2.\ell$
and $n_1.Z = n_2.Z$, since the number of location and possible zones is
finite. Without loss of generality, we can assume that $n_1$ is an ancestor of
$n_2$ and $n_2$ is the parent of another node. This is only possible if
$n_2$ is in \passed, which means that $n_2$ was in \wait and was
not covered at some point. Since a zone can only be modified to be
smaller, this means that $n_2$ has been modified at some
point. Otherwise $n_2$ has always been covered by~$n_1$, which is not
possible. Since $n_2$ has been modified, and it is covered (at~least
by~$n_1$), this means that \cut has been called on~$n_2$ last
time its zone has been modified. This means that $n_2$ has no children
and is not in \passed, contradicting our assumption.
Hence the algorithm always terminates.
\end{proof}

\section{Symbolic Algorithm}
\label{sec-symb}

\subsection{Boolean encoding of zones}
We now present a symbolic algorithm that represents abstract states
using Boolean formulas.  Let $\bbB=\{0,1\}$, and~$\calV$ be a set of
variables.  A~Boolean formula~$f$ that uses variables from
set~$X\subseteq\calV$ will be written~$f(X)$ to make the dependency
explicit; we~sometimes write~$f(X,Y)$ in place~of~$f(X\cup Y)$.
Such~a formula represents a set $\sem{f} = \{v\in
\bbB^{\calV} \mid v\models f\}$.  We~consider primed versions of all
variables; this~will allow us to write formulas relating two valuations.
For~any subset~$X\subseteq\calV$, we~define $X' = \{p' \mid p \in X\}$.

A~\emph{literal} is either~$p$ or $\lnot p$ for a variable~$p$.  Given
a set~$X$ of variables, an~\emph{$X$-minterm} is the conjunction of
literals where each variable of~$X$ appears exactly once. $X$-minterms can
be seen as elements of~$\mathbb{B}^X$.

Given a vector of Boolean formulas~$Y=(Y_x)_{x \in X}$,
formula~$f[Y/X]$ is the \emph{substitution of~$X$ by~$Y$ in~$f$},
obtained by replacing each~$x \in X$ with the formula~$Y_x$.
The~positive cofactor of~$f(X)$ by~$x$ is~$ \exists x.\ (x \land
f(X))$, and its negative cofactor is~$\exists x.\ (\lnot x \land
f(X))$.

Let us define a generic operator~$\postop$ that computes successors of a
set~$S(X,Y)$ given a relation~$R(X,X')$ (here, $Y$~designates any set
of variables on which~$S$ might depend outside of~$X$):
\(
\postop_R(S(X,Y)) = (\exists X. S(X,Y) \land R(X,X'))[X/X']
\).
Similarly, we~set
\(
\preop_R(S(X,Y)) = (\exists X'. S(X,Y)[X'/X] \land R(X,X'))
\),
which computes the predecessors of~$S(X,Y)$ by the relation~$R$~\cite{McM93}.

\paragraph{Clock predicate abstraction.}
We~fix a total order~$\avant$ on~$\Clocksz$. 
%
In this section, abstract domains are defined as $\Dom =
(\Dom_{x,y})_{{x\avant y} \in \Clocksz}$, that is only for pairs $x\avant y$. 
In~fact, constraints of the form~$x-y \leq k$ with $x\apres y$ are
encoded using the negation of $y-x<-k$ since $(x-y\leq k)
\Leftrightarrow \lnot(y-x <-k$).  We thus define~$\Dom_{x,y} =
-\Dom_{y,x}$ for all~$x \apres y$.


For~$x,y \in \Clocksz$, let~$\Pred_{x,y}$ denote the set of
\emph{clock predicates associated to~$\Dom_{x,y}$}:
\[
	\Pred_{x,y}^\Dom = \{ \clPred{x-y\prec k}\mid (k,\prec) \in \Dom_{x,y}\}.
\]
Let~$\Pred^\Dom = \cup_{x,y \in \Clocksz} \Pred_{x,y}$ denote the set
of all clock predicates associated with~$\Dom$ (we~may omit the
superscript~$\Dom$ when it is clear).  For all~$(x,y) \in \Clocksz^2$
and~${(k,\prec) \in \Dom_{x,y}}$, we~denote by~$\clpred{x-y\prec k}$
the literal $\clPred{x-y\prec k}$ if~$x\avant y$, and~$\lnot \clPred{y-x
  \prec^{-1} -k}$ otherwise (where $\mathord\leq^{-1}=\mathord<$ and
$\mathord<^{-1}=\mathord\leq$).
We also consider a set~$\Bits$ of Boolean variables used to encode
locations.  Overall, the state space is described using Boolean
formulas on these two types of variables, so states are elements of
$\bbB^{\Pred\cup\Bits}$.

Our Boolean encoding of clock constraints and semantic operations
follow those of~\cite{SB-cav03} for a concrete
domain. We~define these however for abstract domains, and show how all
successor computation and refinement operations can be performed.

Let us define the \emph{clock semantics} of predicate~$\clPred{x-y\preceq
  k}$ as $\sem{\clPred{x-y \preceq k}}_{\Clocksz} = \{\nu \in
\mathbb{R}_{\geq 0}^\Clocksz \mid \nu(x) - \nu(y) \preceq k\}$.  Since
the set~$\Clocks$ of clocks is fixed, we~may omit the subscript and
just write $\sem{\clPred{x-y\preceq k}}$.
We~define the conjunction, disjunction, and negation as intersection,
union, and complement, respectively.
Given a $\Pred$-minterm~$v \in \bbB^\Pred$,
we define~$\sem{v}_\Dom = \bigcap_{p \text{ s.t. } v(p)} \sem{p}_\Dom \cap \bigcap_{p \text{ s.t. } \lnot v(p)} \sem{p}_\Dom^c$.
Thus, negation of a predicate encodes its complement.
For a Boolean formula~$F(\Pred)$, we~set
$\sem{F} = \bigcup_{v \in \minterms(F)} \sem{v}_\Dom$.
Intuitively, the minterms of~$\Pred$ define smallest zones of~$\Realnn^\Clocks$ definable using~$\Pred$.
A~minterm $v \in \bbB^{\calB\cup\Pred}$ defines a pair
$\sem{v}_\Dom = (l,Z)$ where~$l$~is encoded by~$v_{|\calB}$ and~$Z=\sem{v_{|\Pred}}_\Dom$.
A~Boolean formula~$F$ on~$\calB \cup \Pred$ defines a set 
$\sem{F}_\Dom = \cup_{v \in \minterms(F)} \sem{v}_\Dom$ of such pairs.
A~minterm~$v$ is \emph{satisfiable} if~$\sem{v}_\Dom \neq \emptyset$.


An abstract domain~$\Dom$ induces an \emph{abstraction 
function}~$\alpha_\Dom\colon {2^{\Realnn^\Clocks}} \rightarrow 2^{\bbB^\Pred}$
with~$\alpha_\Dom(Z) = \{ v \mid v \in \bbB^\Pred\text{ and }
\sem{v}_\Dom \cap Z \neq \emptyset\}$,
from the set of zones to the set of subsets of Boolean valuations on~$\Pred$.
We~define the \emph{concretization function} as~$\sem{\cdot}_\Dom \colon 2^{\bool^\Pred} \rightarrow {2^{\Realnn^\Clocks}}$.
The~pair~$(\alpha_\Dom, \sem{\cdot}_\Dom)$ is a Galois connection, and
$\sem{\alpha_\Dom(Z)}_\Dom$ is the most precise abstraction of~$Z$ in the domain induced by~$\Dom$.
Notice that $\alpha_\Dom$ is non-convex in general:
for instance, if the clock predicates are~$x\leq 2,y\leq 2$, then the
set defined by the constraint~$x=y$ maps to $(\clpred{x\leq 2} \land
\clpred{y\leq 2}) \lor (\lnot \clpred{x\leq 2} \land \lnot
\clpred{y\leq 2})$, which is not convex.



\subsection{Reduction} 
\label{section:reduction}

We now define the reduction operation, which is similar to the
reduction of DBMs.  The~idea is to eliminate unsatisfiable minterms
from a given Boolean formula.  For~example, we~would like to make sure
that in all minterms, if~$p_{x-y\leq 1}$ holds, then so
does~$p_{x-y\leq 2}$, when both are available predicates.  Another
issue is to eliminate minterms that are unsatisfiable due to triangle
inequality.  This is similar to the shortest path computation used to
turn DBMs in canonical form.

Let a \emph{path} in~$\Dom$ be a sequence
$x_1,(\alpha_1,\mathord{\prec_1}),x_2,(\alpha_2,\mathord{\prec_2}), \ldots,
x_k,(\alpha_k,\mathord{\prec_k}),x_{k+1}$ where~$x_1,\ldots,x_{k+1} \in
\Clocksz$, and~$(\alpha_i,\mathord{\prec_i}) \in \Dom_{x_i,x_{i+1}}$ for $1\leq
i\leq k$.  Let us define $\Paths_k^\Dom(x-y\prec \alpha)$ as the set
of paths from~$x$ to~$y$ of length~$k$ and weight at
most~$(\alpha,\mathord\prec)$, that~is,
paths~$x_1,(\alpha_1,\mathord{\prec_1}),x_2,(\alpha_2,\mathord{\prec_2}), \ldots,
x_k,(\alpha_k,\mathord{\prec_k}),x_{k+1}$ with~$x_1=x$, $x_{k+1} = y$,
and~$\sum_{i=1}^k (\alpha_i,\mathord{\prec_i}) \leq (\alpha,\mathord\prec)$.
We~also
denote~${\Paths_{\leq k}^\Dom(x-y\prec \alpha)} = \bigcup_{l\leq k}
\Paths_l^\Dom(x-y\prec \alpha)$.  For a
path~$\pi=x_1,(\alpha_1,\mathord{\prec_1}),\ldots,x_{k+1}$ and minterm~$v$, let
us write~$v \models \pi$ for the statement $v \models \wedge_{i=1}^k
\clpred{x_i - x_{i+1}\prec \alpha_i}$.

A~minterm~$v \in \bbB^{\Pred}$ is \emph{$k$-reduced} if for all 
$(x,y) \in \Clocksz^2$ and~$(\alpha,\mathord\prec) \in \Dom_{x,y}$, 
for all~$\pi = x_1,{(\alpha_1,\mathord{\prec_1})},x_2, \ldots,x_{k+1} \in \Paths_{\leq k}(x-y\prec \alpha)$,
whenever $v \models \pi$, we also have $v \models \clpred{x_1 -
  x_{k+1}\prec \alpha}$ for all~$(\alpha,\mathord\prec) \in
\Dom_{x_1,x_{k+1}}$ with $(\alpha,\mathord\prec) \geq \sum_{i=1}^k
(\alpha_i,\mathord{\prec_i})$.
Thus, in a~$k$-reduced minterm, no~contradictions can be obtained via
paths of length less than or equal to~$k$.  Observe that
$|\Clocksz|$-reduction is equivalent to satisfiability since the
condition then includes all paths, and it is known that in the absence
of negative cycles, a~set of difference constraints is
satisfiable. Furthermore, for the concrete domain, $k$-reduction is
equivalent to reduction for any~$k\geq 2$.  A formula is said to be
$k$-reduced if all its minterms are $k$-reduced.

\begin{example}
  \label{example:reduced}
  Given predicates $\Pred = \{\clpred{x-y\leq 1}, \clpred{y-z\leq 1},
  \clpred{x-z\leq 2}\}$, the formula $\clpred{x-y\leq 1} \land
  \clpred{y-z\leq 1}$ is not reduced since it contains the
  unsatisfiable minterm $\clpred{x-y\leq 1} \land \clpred{y-z\leq 1}
  \land \lnot \clpred{x-z\leq 2}$.  However, the same formula is
  reduced if~$\Pred = \{\clpred{x-y\leq 1}, \clpred{y-z\leq 1}\}$.

  Consider now the predicate set~$\calP = \{\clpred{x-y\leq 1},
  \clpred{y-z\leq 1}, \clpred{z-w\leq 3}, \clpred{x-w\leq 5}\}$, and
  consider the formula~$\phi = \clpred{x-y\leq 1}\land \clpred{y-z\leq
    1} \land \clpred{z-w\leq 3}$ which is~$2$-reduced.  Notice that
  the reduction of~$\phi$ is $\clpred{x-y\leq 1}\land \clpred{y-z\leq
    1} \land \clpred{z-w\leq 3} \land \clpred{x-w\leq 5}$ since the
  last predicate is implied by the conjunction of others.  However,
  $\phi$ is $2$-reduced since no path of length~$2$ allows to
  deduce~$\clpred{x-w \leq 5}$.  In~the concrete domain,
  $\clpred{x-y\leq 1} \land \clpred{y-z \leq 1}$ would imply
  $\clpred{x-z\leq 2}$, thus, $\clpred{x-w\leq 5}$ could be derived~too.
  It~is indeed because of the abstract domain that $2$-reduction
  might fail to capture all shortest paths.
\end{example}

In this paper, we only consider $2$-reduction since computing
reductions is the most expensive operation in our algorithms, and the
formula below defining $2$-reduction already tends to grow in~size.
Let us define $\reduce_\Dom^2$ as follows
\[
  \bigwedge_{\substack{(x,y) \in \Clocksz^2\\ (k,\mathord{\prec}) \in \Dom_{x,y}}}
  \Biggl[\clpred{x-y \prec k} \leftarrow
  \Bigl(
    \bigvee_{
      \substack{
        (l_1,\mathord{\prec_1}) \in \Dom_{x,y}\\
        (l_1,\mathord{\prec_1}) \leq (k,\mathord{\prec})
      }
    } \clpred{x-y\prec_1 l_1} \lor \!\!\!\!
    \bigvee_{
      \substack{z \in \Clocksz,
        (l_1,\mathord{\prec_1}) \in \Dom_{x,z}, \\
        (l_2,\mathord{\prec_2}) \in \Dom_{z,y}\\
        (l_1,\mathord{\prec_1})+(l_2,\mathord{\prec'_2})\leq (k,\mathord{\prec})
      }
    }
    \!\!\!\!
    \clpred{x-z\prec_1 l_1} \mathrel{\wedge} \clpred{z-y\prec_2 l_2}
  \Bigr)\Biggr]
\]
The formula intuitively applies shortest paths over
paths of length $1$ or~$2$.


\begin{restatable}{lemma}{lemmareduce}
\label{lemma:reduce}
  For all formulas~$S(\Pred)$,
  we have $\sem{S}_\Dom = \sem{\reduce_\Dom^2(S)}_\Dom$
  and all minterms of~$\reduce_\Dom^2(S)$ are $2$-reduced.
\end{restatable}


\begin{proof}
  It is easy to see that $\sem{\reduce_\Dom^k(S)}_\Dom \subseteq \sem{S}_\Dom$.
  In fact, any minterm of the former is a minterm of~$S$ as well by definition.

  To see the converse, consider~$v \in \minterms(S)$. We show
  that~$\sem{v}_\Dom \subseteq \sem{\reduce_\Dom^k(S)}_\Dom$.
  If~$\sem{v}_\Dom = \emptyset$ then the inclusion holds trivially.
  Otherwise, we must have $v \in \minterms(\reduce_\Dom^k(S))$. In
  fact, we show that all minterms~$v \not \in
  \minterms(\reduce_\Dom^k(S))$ satisfy~$\sem{v} = \emptyset$.
  Consider such a minterm~$v$. There must exist~$x,y \in \Clocksz$ and
  $(k,\prec_k) \in \Dom_{x,y}$ such that~$\lnot \clpred{x-y\prec k}$ but
  the right hand side of the implication holds.  But this implies
  that~$\sem{v} = \emptyset$.

  The fact that all miterms of $\reduce_\Dom^k(S)$ are $k$-reduced
  follow by the definition of the operator.
\end{proof}

\begin{example}
  Note that~$\reduce_\Dom^2(S)$ can still contain valuations that are
  unsatisfiable.  Consider $\Pred=\{ \clpred{x-y\leq 1}, \clpred{y-z\leq 1},
  \clpred{x-z\leq 4}, \clpred{z-x\leq -3}\}$.  Then the minterm~$u$ that sets
  all predicates to true is still contained in $\reduce_\Dom^2(S)$
  although~$\sem{u} = \emptyset$ since~$x-y\leq 1 \land y-z\leq 1$
  implies~$x-z\leq 2$ which contradicts~$z-x\leq -3$.  Here, adding
  the predicate~$\clpred{x-z\leq 2}$ or~$\clpred{x-z<3}$ would render the
  abstraction precise enough to eliminate this valuation
  in~$\reduce_\Dom^2(S)$.
\end{example}

Let us see how an abstraction can be refined so that the reduced constraint
eliminates a given unsatisfiable minterm.

\begin{restatable}{lemma}{lemmaemptyrefine}
\label{lemma:empty-refine}
  Let~$v \in \bbB^{\Pred^\Dom}$ be a minterm such
  that~$v \models \reduce_\Dom^2$ and
  $\sem{v} = \emptyset$.
  One~can compute in polynomial time a refinement~$\Dom' \supset \Dom$
  such that~$v \not \models \reduce_{\Dom'}^2$.
\end{restatable}


\begin{proof}
  Consider a (non-canonial) DBM~$D$ that encodes~$v$. Formally, for
  all~$(x,y) \in \Clocksz^2$, $D(x,y) = \min \{ (k,\mathord\prec)
  \mid \clpred{x - y \prec k} \in \Pred_{x,y}\text{ and } v(\clpred{x - y \prec
    k}) = 1\}$.  The~corresponding graph must have a negative
  cycle~$s_1s_2\ldots s_m$ with $s_m = s_1$.  To~define~$\Dom'$,
  we~add the following predicates to~$\Dom$:
  \begin{itemize}
    \item $s_i - s_{i+1} \leq D(s_i,s_{i+1})$ for all~$1\leq j \leq m-1$.
    \item $s_1 - s_{j} \leq D(s_1,s_2) + D(s_2,s_3) + \ldots + D(s_{j-1},s_{j})$
      for~$1\leq j\leq m-1$.
    \end{itemize}
    Intuitively, along the negative cycle, we are adding predicates to
    represent exactly each single step, and also each big step from
    $s_1$ to~$s_j$. This allows to derive the negative cycle using
    only paths of length~$2$.

    More precisely, $v \land \reduce^2_{\Dom'}$ implies the following
    two predicates: $s_1 - s_m \leq \sum_{j=1}^{m-1} D(s_j,s_{j+1})$,
    and $s_m - s_1 \leq D(s_m,s_1)$ as implied by~$v$.  Since
    $-D(s_m,s_1) > \sum_{j=1}^{m-1} D(s_j,s_{j+1})$ (due to the
    negative cycle), the~first predicate entails that $\lnot \clpred{s_1 -
      s_m \leq -D(s_m,s_1)}$, which contradicts the second one. Hence~$v
    \not \models \reduce^2_{\Dom'}$.
\end{proof}

\subsection{Successor Computation}
In this section, we explain how successor computation is realized in
our encoding.
For a guard~$g$, assume we have computed an
abstraction~$\alpha_\Dom(g)$ in the present abstract domain. 
%
For each transition~$\sigma=(\ell_1,g,R,\ell_2)$, let us define the
formula~$T_\sigma = \ell_1 \land \alpha_\Dom(g)$.  We~show how each
basic operation on zones can be computed in our BDD encoding.
In~our algorithm, all formulas~$A(\Bits,\Pred)$ representing sets of
states are assumed to be reduced, that~is, $A(\Bits,\Pred) \subseteq
\reduce_\Dom^2(A(\Bits,\Pred))$.


\medskip
The~intersection operation is simply logical conjunction.

\begin{restatable}{lemma}{lemmainter}
\label{lemma:inter}
  For all reduced formulas~$A(\Pred), B(\Pred)$,
  $A(\Pred) \land B(\Pred) = \alpha_\Dom(\sem{A(\Pred)}_\Dom \cap \sem{B(\Pred)}_\Dom)$.
\end{restatable}

\begin{proof}
  Consider~$v \in \minterms(A(\Pred) \land B(\Pred))$.  Then $v \in
  \minterms(A(\Pred)) \cap \minterms(B(\Pred))$.  So~$\sem{v}
  \subseteq \sem{A(\Pred)} \cap \sem{B(\Pred)}$, and~$v =
  \alpha_\Dom(\sem{v}) \subseteq \alpha_\Dom(\sem{A(\Pred)}_\Dom \cap
  \sem{B(\Pred)}_\Dom)$.

  We will now show $\alpha_\Dom(\sem{A(\Pred)}_\Dom \cap
  \sem{B(\Pred)}_\Dom) \subseteq A(\Pred) \land B(\Pred)$, which is
  equivalent to $\sem{A(\Pred)}_\Dom \cap \sem{B(\Pred)}_\Dom
  \subseteq \sem{A(\Pred) \land B(\Pred)}$.  Consider any clock
  valuation~$\nu$ in the LHS, and let~$v = \alpha_\Dom(\nu)$.
  Since~$\nu \in \sem{A(\Pred)}_\Dom$ and~$\nu \in
  \sem{B(\Pred)}_\Dom$, we must have $v \in A(\Pred) \land B(\Pred)$,
  so $\nu \in \sem{A(\Pred) \land B(\Pred)}$.
\end{proof}

\medskip
For the time successors, we define
\[
S_{\Up} = 
 \bigwedge_{\substack{x \in \Clocks\\ (k,\prec) \in \Dom_{x,0}}}
 (\lnot \clpred{x-0\prec k} \rightarrow \lnot \clprpr{x-0\prec k})
 \bigwedge_{\substack{x,y \in \Clocksz, x\neq 0\\ (k,\prec) \in \Dom_{x,y}}}
 (\clprpr{x-y\prec k} \leftrightarrow \clpred{x-y \prec k}).
 \]
Note that this relation is not a function: for~$x\leq 0,
y=0$, if~$\lnot \clpred{x-y\prec k}$, then necessarily $\lnot
\clprpr{x-y\prec k}$; but otherwise both truth values for~$\clprpr{x-y\prec
  k}$ are allowed.  In~fact, the~formula only says that all lower
bounds on clocks and diagonal constraints must be preserved.
We~let $\Up(A(\Bits,\Pred)) = \reduce(\postop_{S_\Up}(A(\Bits,\Pred)))$.

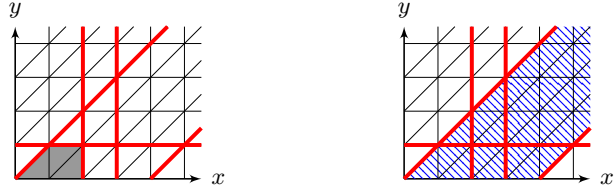
\begin{figure}[t]
  \centering
  \begin{tikzpicture}
    \begin{scope}[scale=.45]
      \fill[black!40] (0,0) -- (2,0) -- (2,1) -- (1,1)-- cycle;
      \draw[-latex',line width=.6pt] (0,0) -- (5.5,0);
      \draw[-latex',line width=.6pt] (0,0) -- (0,4.5);
      \node at (0,5) {$y$};
      \node at (6,0) {$x$};
      \begin{scope}[line width=0.1pt]
        \draw[-] (1,0) -- (1,4.5);
        \draw[-] (2,0) -- (2,4.5);
        \draw[-] (3,0) -- (3,4.5);
        \draw[-] (4,0) -- (4,4.5);
        \draw[-] (5,0) -- (5,4.5);
        \draw[-] (0,1) -- (5.5,1);
        \draw[-] (0,2) -- (5.5,2);
        \draw[-] (0,3) -- (5.5,3);
        \draw[-] (0,4) -- (5.5,4);
        \draw[-] (0,0) -- (4.5,4.5);
        \draw[-] (1,0) -- (5.5,4.5);
        \draw[-] (2,0) -- (5.5,3.5);
        \draw[-] (3,0) -- (5.5,2.5);
        \draw[-] (0,1) -- (3.5,4.5);
        \draw[-] (0,2) -- (2.5,4.5);
        \draw[-] (0,3) -- (1.5,4.5);
      \end{scope}
      \begin{scope}[line width=1.5pt,color=red]
        \draw[-] (0,1) -- (5.5,1);
        \draw[-] (2,0) -- (2,4.5);
        \draw[-] (3,0) -- (3,4.5);
        \draw[-] (0,0) -- (4.5,4.5);
        \draw[-] (4,0) -- (5.5,1.5);
      \end{scope}
    \end{scope}
  \end{tikzpicture}
\hfil
  \begin{tikzpicture}
    \begin{scope}[scale=.45]
      \draw[-latex',line width=.6pt] (0,0) -- (5.5,0);
      \draw[-latex',line width=.6pt] (0,0) -- (0,4.5);
      \fill[pattern=north west lines, pattern color=blue]
      (0,0) -- (4.5,4.5) -- (5.5,4.5) -- (5.5,1.5) -- (4,0) -- cycle;
      \node at (0,5) {$y$};
      \node at (6,0) {$x$};
      \begin{scope}[line width=0.1pt]
        \draw[-] (1,0) -- (1,4.5);
        \draw[-] (2,0) -- (2,4.5);
        \draw[-] (3,0) -- (3,4.5);
        \draw[-] (4,0) -- (4,4.5);
        \draw[-] (5,0) -- (5,4.5);
        \draw[-] (0,1) -- (5.5,1);
        \draw[-] (0,2) -- (5.5,2);
        \draw[-] (0,3) -- (5.5,3);
        \draw[-] (0,4) -- (5.5,4);
        \draw[-] (0,0) -- (4.5,4.5);
        \draw[-] (1,0) -- (5.5,4.5);
        \draw[-] (2,0) -- (5.5,3.5);
        \draw[-] (3,0) -- (5.5,2.5);
        \draw[-] (0,1) -- (3.5,4.5);
        \draw[-] (0,2) -- (2.5,4.5);
        \draw[-] (0,3) -- (1.5,4.5);
      \end{scope}
      \begin{scope}[line width=1.5pt,color=red]
        \draw[-] (0,1) -- (5.5,1);
        \draw[-] (2,0) -- (2,4.5);
        \draw[-] (3,0) -- (3,4.5);
        \draw[-] (0,0) -- (4.5,4.5);
        \draw[-] (4,0) -- (5.5,1.5);
      \end{scope}
    \end{scope}
  \end{tikzpicture}
  \caption{Time successors: the successors of the gray zone on the left are the
  union of all blue dashed zones on the right.}
  \label{fig:up}
\end{figure}

\begin{example}
Figure~\ref{fig:up} shows this operation applied to the gray zone on the left.
The overapproximation is visible in this example. In~fact,
the blue dashed zone defined by $Z = 3< x< 5 \land 0\leq y< 1 \land x-y< 4$
(on the bottom right) is computed as a successor of the gray zone
although no point of the gray zone can actually reach~$Z$ by a time delay.
Adding more diagonal constraints to the abstract domain, for instance, 
$x-y\leq 2$ would eliminate this successor.
\end{example}

\begin{restatable}{lemma}{lemmaup}
\label{lemma:up}
	For any Boolean formula~$A(\Bits, \Pred)$,
    $\alpha_\Dom(\timesucc{\sem{A}})\subseteq \Up(A)$.
    Moreover, if~$\Dom$ is the concrete domain and $A$ is reduced, 
    then this holds with equality.    
\end{restatable}

\begin{proof}
  We show that 
  $\alpha_\Dom(\timesucc{\sem{A}})\subseteq \Up(A)$,
  which is equivalent to 
  $\timesucc{\sem{A}}\subseteq \sem{\Up(A)}_\Dom$.
  Let~$\nu' \in \timesucc{\sem{A}}$, and let~$\nu \in \sem{A}$ such that
  $\nu' = \nu + d$ for some~$d\geq 0$.
  We write~$u = \alpha_\Dom(\nu)\in \minterms(A)$, and~$v = \alpha_\Dom(\nu')$.
  We now show that~$(u,v) \in S_\Up$. This suffices to prove the inclusion since
  $v \in \minterms(\Up(A))$ implies that~$\nu' \in \sem{v}_\Dom \subseteq \sem{\Up(A)}_\Dom$.
  Observe that $\alpha_\Dom(\nu)$ and~$\alpha_\Dom(\nu')$ satisfy the same diagonal predicates
  of the form~$p_{x-y\prec \alpha}$ with~$x,y \neq 0$, since~$\nu' = \nu+d$.
  Moreover, any lower bound satisfied by~$u$ is also satified by~$v$. Thus, $(u,v) \in S_\Up$.

  \smallskip
  Assume now that~$A$ is reduced. We show that $\Up(A) \subseteq \alpha_\Dom(\timesucc{\sem{A}})$.
  Let~$v \in \minterms(\Up(A))$, and let~$(u,v) \in S_\Up$ with~$u \in \minterms(A)$.
  We claim that~$\sem{v} \subseteq \timesucc{\sem{u}}$.
  This suffices to prove the inclusion since 
  $\{v\} = \alpha_\Dom(\sem{v}) \subseteq \alpha_\Dom(\timesucc{\sem{u}}) \subseteq \alpha_\Dom(\timesucc{\sem{A}})$.
  To~prove the claim, it~suffices to see~$\sem{u}$ and~$\sem{v}$ as DBMs. In fact, $S_\Up$ precisely corresponds to the up operation
  on DBMs. In particular, $u$ and~$v$ have the same diagonal constraints, and any lower bound of~$u$ is a lower bound of~$v$.
  Because~$u$ is reduced, this implies that~$\timesucc{\sem{v}}
  \subseteq \sem{u}$.
\end{proof}

\begin{example}
  Note that we do need that the domain is concrete to prove
  equality. In~fact, assume the predicates are $p_1 = y-x\leq 3$,
  $p_2= x-y\leq 1$, $x\leq 1, y\leq 2, x\leq 2, y\leq 4$.  Then if
  $u=x\leq 1\land y\leq 2\land p_1\land p_2$ and~$v=1\leq x\leq 2\land
  y\leq 4\land p_1\land p_2$.  So~$v$ is larger than $\timesucc{u}$.

  Second, to see that the $\Up$ operation can yield strict
  over-approximations, consider $\clpred{x\leq 1} \land \clpred{y\geq
    1}$ in a domain with no predicate on~$x-y$.  The operator relaxes
  all upper bounds, yielding~$\clpred{y\geq 1}$ which defines a set
  larger than  the concrete
  time-successors $y\geq 1 \land x-y \leq 0$.
\end{example}

Alternatively, we can also use the method described in~\cite[Theorem
  2]{SB-cav03} to compute time successors, but the above relation will
allow us to compute predecessors as well.

\medskip
Last, we define the reset operation as follows.
For any~$z \in \Clocks$,
\begin{xalignat*}1
  S_{\Reset_z} ={} &\phantom{{}\et{}}
  \bigwedge_{(0,\leq) \leq (k,\mathord{\prec}) \in \Dom_{z,0}}{\clprpr{z-0\prec k}}\\
  &{} \et \bigwedge_{\substack{x \neq z\\(k,\prec) \in \Dom_{x,z}}}
  \Bigl( \bigvee_{\substack{(l,\prec') \in \Dom_{x,0} \\ (l,\prec') \leq (k,\prec) }}
  \clpred{x-0\prec' l} \Bigr)
  \Rightarrow \clprpr{x-z\prec k}
    \\\noalign{\allowbreak}
  &{}\et
  \bigwedge_{\substack{y \neq z  \\ (k,\prec) \in \Dom_{z,y}}}
  \Bigl( \bigvee_{\substack{(l,\prec') \in \Dom_{0,y}\\ (l,\prec') \leq (k,\prec)}}\clpred{0 - y \prec' l} \Bigr)
  \Rightarrow \clprpr{z-y\prec k}\\\noalign{\allowbreak}
  &{}\et
  \bigwedge_{\substack{y \neq z \\
      (0,\leq) \leq (k,\prec) \in \Dom_{z,y}}} \clprpr{z-y\prec k}\\
  &{}\et
    \bigwedge_{\substack{x, y \neq z \\ (k,\prec) \in \Dom_{x,y}}}
    \clprpr{x-y\prec k} \Leftrightarrow p_{x-y\prec k}.
  \end{xalignat*}
Intuitively, the first conjunct ensures all non-negative upper bounds
on the reset clock hold; the~second conjunct ensures that a diagonal
predicate $x-z\prec k$ with~$x \neq z$ is set to true if, and only if,
an upper bound~$(l,\prec') \leq (k,\prec)$ already holds on~$x$.
Recall that in operations on~DBMs, one~sets such a diagonal component
to the tightest upper bound on~$x$.  The~third conjunct is symmetric
to the second, and the last one ensures diagonals not affected by
reset are unchanged.
Let us define~$\Reset_z(A) = \reduce(\postop_{S_{\Reset_z}}(A))$.

\begin{figure}[t]
  \centering
  \begin{tikzpicture}
    \begin{scope}[scale=.45]
      \fill[black!40] (0,1) -- (2,3) -- (0,3) -- cycle;
      \draw[-latex',line width=.6pt] (0,0) -- (5.5,0);
      \draw[-latex',line width=.6pt] (0,0) -- (0,4.5);
      \node at (0,5) {$y$};
      \node at (6,0) {$x$};
      \begin{scope}[line width=0.1pt]
        \draw[-] (1,0) -- (1,4.5);
        \draw[-] (2,0) -- (2,4.5);
        \draw[-] (3,0) -- (3,4.5);
        \draw[-] (4,0) -- (4,4.5);
        \draw[-] (5,0) -- (5,4.5);
        \draw[-] (0,1) -- (5.5,1);
        \draw[-] (0,2) -- (5.5,2);
        \draw[-] (0,3) -- (5.5,3);
        \draw[-] (0,4) -- (5.5,4);
        \draw[-] (0,0) -- (4.5,4.5);
        \draw[-] (1,0) -- (5.5,4.5);
        \draw[-] (2,0) -- (5.5,3.5);
        \draw[-] (3,0) -- (5.5,2.5);
        \draw[-] (0,1) -- (3.5,4.5);
        \draw[-] (0,2) -- (2.5,4.5);
        \draw[-] (0,3) -- (1.5,4.5);
      \end{scope}
      \begin{scope}[line width=1.5pt,color=red]
        \draw[-] (0,1) -- (3.5,4.5);
        \draw[-] (0,3) -- (5.5,3);
        \draw[-] (4,0) -- (4,4.5);
        \draw[-] (0,1) -- (5.5,1);
        \draw[-] (2,0) -- (5.5,3.5);
      \end{scope}
    \end{scope}
    \begin{scope}[scale=.45,shift={(10,0)}]
     \fill[pattern=north west lines, pattern color=blue] (0,0) -- (4,0) -- (4,1) -- (0,1) -- cycle;
     \draw[-latex',line width=.6pt] (0,0) -- (5.5,0);
      \draw[-latex',line width=.6pt] (0,0) -- (0,4.5);
      \node at (0,5) {$y$};
      \node at (6,0) {$x$};
      \begin{scope}[line width=0.1pt]
        \draw[-] (1,0) -- (1,4.5);
        \draw[-] (2,0) -- (2,4.5);
        \draw[-] (3,0) -- (3,4.5);
        \draw[-] (4,0) -- (4,4.5);
        \draw[-] (5,0) -- (5,4.5);
        \draw[-] (0,1) -- (5.5,1);
        \draw[-] (0,2) -- (5.5,2);
        \draw[-] (0,3) -- (5.5,3);
        \draw[-] (0,4) -- (5.5,4);
        \draw[-] (0,0) -- (4.5,4.5);
        \draw[-] (1,0) -- (5.5,4.5);
        \draw[-] (2,0) -- (5.5,3.5);
        \draw[-] (3,0) -- (5.5,2.5);
        \draw[-] (0,1) -- (3.5,4.5);
        \draw[-] (0,2) -- (2.5,4.5);
        \draw[-] (0,3) -- (1.5,4.5);
      \end{scope}
      \begin{scope}[line width=1.5pt,color=red]
        \draw[-] (0,1) -- (3.5,4.5);
        \draw[-] (2,0) -- (5.5,3.5);
        \draw[-] (0,3) -- (5.5,3);
        \draw[-] (4,0) -- (4,4.5);
        \draw[-] (0,1) -- (5.5,1);
      \end{scope}
    \end{scope}

  \end{tikzpicture}
  \caption{Reset Operation: The abstract successor of the gray zone on the left is the blue dashed zone on the right.}
  \label{fig:reset}
\end{figure}
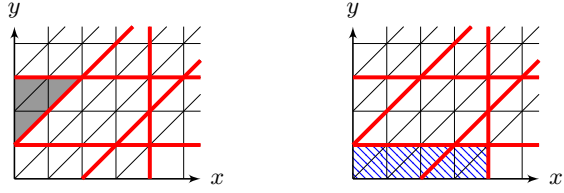

\begin{restatable}{lemma}{lemmaopreset}
\label{lemma:op-reset}
  For any Boolean formula~$A(\Bits,\Pred)$, and for any~$z\in \Clocks$, we
  have $\alpha_\Dom(\zReset_z(\sem{A}_\Dom)) \subseteq \Reset_z(A)$.
  Moreover, if~$\Dom$ is the concrete domain, and~$A$ is reduced, then
  the above holds with equality.
\end{restatable}

\begin{proof}
  Let $r = \{z\}$.  We show the equivalent inclusion
  $\zReset_r(\sem{A}_\Dom) \subseteq \sem{\Reset_r(A)}_\Dom$.
  Let~$\nu' \in \zReset_r(\sem{A})$, and $\nu \in \sem{A}$ such
  that~$\nu' = \nu[r\leftarrow 0]$.  So~$u = \alpha_\Dom(\nu) \in
  A$. Letting~$v = \alpha_\Dom(\nu')$, let us show that~$(u,v) \in
  S_{\Reset_r}$, which proves that~$v \in \Reset_r(A)$, thus~$\nu' \in
  \sem{v}_\Dom \subseteq \sem{\Reset_r(A)}_\Dom$.

  \begin{itemize}
  \item Consider~$x \in r$. Since~$\nu'(x) = 0$, we have that~$v
    \models \clpred{x-y\prec k}$ for all $(k,\prec) \in \Dom_{x,y}$
    with~$(k,\mathord\prec) \geq (0,\mathord\leq)$;
    so $(u,v)$ satisfies the first conjunct.
    Fix~$(k,\mathord\prec) \in \Dom_{x,y}$.

  \item Assume~$x \not \in r, y \in r$, so that~$\nu'(x) = \nu(x),
    \nu'(y) = 0$.  If there is~$(l,\mathord{\prec'}) \in \Dom_{x,0}$
    with~$(l,\mathord{\prec'}) \leq (k,\mathord\prec)$, this
    means~$\nu(x) - 0 \prec' l \prec k$ so~$\nu'(x) - \nu'(y) \prec
    k$.  Thus the second conjunct is satisfied.

  \item Assume~$x \in r, y \not \in r$ so that~$\nu'(x) = 0, \nu'(y) =
    \nu(y)$.  If there is~$(l,\mathord{\prec'}) \in \Dom_{0,y}$
    with~$(l,\mathord{\prec'}) \leq (k,\mathord\prec)$, this means~$0
    - \nu(y) \prec' l \prec k$ so~$\nu'(x) - \nu'(y) \prec k$ as~well.
    Thus the third conjunct is satisfied.  We also have trivially
    $\nu'(x) - \nu'(y) \prec k $ for all~$(0,\mathord\leq) \leq
    (k,\mathord\prec)$, which entails the fourth conjunct.

  \item Observe that $\nu$ and~$\nu'$ satisfy the same predicates of
    type~$\clpred{x-y\prec \alpha}$ with~$x,y \not \in r$ since these
    values are not affected by the reset; so the pair~$(u,v)$
    satisfies the last conjunct of~$S_{\Reset_r}$ as well.
  \end{itemize}


  Now, if the domain is concrete and~$A$ is reduced, then we have
  $\sem{\Reset_r(A)}_\Dom \subseteq \zReset_r(\sem{A}_\Dom)$. In~fact,
  the operation then corresponds precisely to the reset operation in
  DBMs, see~\cite[Algorithm 10]{BY04}.  On~DBMs, for a
  component~$(x,y)$ with~$x \in r, y \not \in r$, the algorithm
  consists in setting this component to value~$(0,y)$. In~our
  encoding, we~thus set the predicate~$\clprpr{x-y\prec k}$ to true
  whenever $\clpred{0-x\prec' l}$ holds with~$(l,\mathord{\prec'})
  \leq (k,\mathord\prec)$.  The argument is symmetric for $(x,y)$
  with~$x \not \in r$, $y \in r$.
\end{proof}


\subsection{Model-checking algorithm}

Algorithm~\ref{alg:reach} shows how to check the reachability of a
target location given an abstract domain.
The list \layers contains, at position~$i$, the set of states that are 
reachable in~$i$ steps.  The function $\ApplyEdges$ computes
the disjunction of immediate successors by all edges.  It~consists in
looping over all edges~$e=(l_1,g,R,l_2)$, and gathering the following
image by~$e$:
\[
   \enc(\ell_2) \land
   \Reset_{r_k}(\Reset_{r_{k-1}}(\ldots(\Reset_{r_1}((((\exists \Bits. A(\Bits,\Pred) \land \enc(\ell_1)) \land \alpha_\Dom(g))))))),
\]
where~$R=\{r_1,\ldots,r_k\}$. We thus use a partitioned transition relation and do not
compute the monolithic transition relation.

When the target location is found to be reachable, \ExtractTrace{}(\layers)
returns
a trace reaching the target location. This is standard and can be done by computing backwards from the last element of \layers, by finding which edge can be applied to reach the current state.
Since both reset and time successor operations are defined using relations, predecessors
in our abstract system can be easily computed using the operator $\preop_R$.
As~it is standard, we~omit the precise definition of this function (the reader can refer to the implementation) but assume that it returns a trace of the form
\[
  A_1 \xrightarrow[]{\sigma_1} A_2 \xrightarrow{\sigma_2} \ldots \xrightarrow{\sigma_{n-1}} A_n,
\]
where the~$A_i(\Bits,\Pred)$ are minterms and the~$\sigma_i$ belong to
the trace alphabet $\Sigma = \{\up, \rempty\} \cup \{r(x)\}_{x \in
  \Clocks}$, with the following meaning:
\begin{itemize}
\item if $A_i\xrightarrow{\up} A_{i+1}$ then $A_{i+1} = \Up(A_i)$;
\item if~$A_i \xrightarrow{\rempty} A_{i+1}$ then $A_{i+1} = A_i$;
\item if~$A_i \xrightarrow{r(x)} A_{i+1}$ then~$A_{i+1} = \Reset_x(A_i)$.
\end{itemize}
The feasibility of such a trace is easily checked using DBMs.
\begin{algorithm}[t]
	\small
	\KwInput{$\TA = (\Locs, \invar, \ell_0, \Clocks, E)$,
          $\ell_T$, $\Dom$}\;
        $\nexts := \enc(l_0) \land \alpha_\Dom(\land_{x \in \Clocks} x=0)$\;
        $\layers := []$\;
        $\reachable := \false$\;
        \While{$(\lnot \reachable \land \nexts) \neq \false$}
        {
          $\reachable := \reachable \lor \nexts$\;
          $\nexts := \ApplyEdges(\Up(\nexts)) \land \lnot \reachable$\;
          $\layers.push(\nexts)$\;
          \uIf{$(\nexts \land \enc(l_T)) \neq \false$}
          {\Return \ExtractTrace(\layers)\;}
        }
        \Return Not reachable\;
        \caption{Algorithm \textsf{SymReach} that checks the
          reachability of a target location~$l_T$ in a given abstract
          domain~$\Dom$.}
    \label{alg:reach}
\end{algorithm}

The overall algorithm then follows a classical CEGAR scheme.
We~initialize~$\Dom$ by
adding the clock constraints that appear syntactically in~$\TA$, which
is often a good heuristic.
We~run the reachability check of Algorithm~\ref{alg:reach}. If no
trace is found, then the target location is not reachable. If a trace
is found, then we check for feasibility.  If it is feasible, then the
counterexample is confirmed. Otherwise, the trace is spurious and we
run the refinement procedure described in the next subsection, and
repeat the analysis.

\subsection{Abstraction refinement}
Since we initialize~$\Dom$ with all clock constraints appearing in
guards, we can
make the following hypothesis.
\begin{assumption}
All
guards are represented exactly in the considered abstractions.
\end{assumption}
Note that the algorithm can be easily extended to the general case;
but this simplifies the presentation.

The abstract transition relation we use is not the most precise
abstraction of the concrete transition relation.  Therefore, it is
possible to have abstract transitions~$A_1 \xrightarrow{a} A_2$ for
some action~$a$ while no concrete transition exists between
$\sem{A_1}$ and~$\sem{A_2}$.  This~requires care and is not a direct
application of the standard refinement technique
from~\cite{CGJ+03}.  A~second difficulty is due to incomplete
reduction of the predicates using~$\reduce^2_\Dom$. In~fact, some
reachable states in our abstract model will be unsatisfiable.  Let us
explain how we refine the abstraction in each of these cases.

Consider an algorithm $\interp$ that returns an interpolant of
two given zones~$Z_1,Z_2$.
In what follows, by the \emph{refinement of~$\Dom$
  by~$\interp(Z_1,Z_2)$}, we~mean the domain~$\Dom'$ obtained by
adding $(k,\mathord\prec)$ to~$\Dom_{x,y}$ for all constraints $x-y\prec k$
of~$\interp(Z_1,Z_2)$.  Observe that~$\alpha_{\Dom'}(Z_1)\cap
\alpha_{\Dom'}(Z_2) = \emptyset$ in this case.

We define concrete successor and predecessor operations for the
actions in~$\Sigma$.  For each~$a \in \Sigma$, let~$\preta_a^c$ denote
the concrete predecessor operation on zones defined straightforwardly,
and similarly for~$\postta^c_a$.

Consider domain~$\Dom$ and the induced abstraction
function~$\alpha_\Dom$.  Assume that we are given a spurious trace
$\pi=A_1 \xrightarrow[]{\sigma_1} A_2 \xrightarrow{\sigma_1} \ldots
\xrightarrow{\sigma_{n-1}} A_n$.  Let~$B_1\ldots B_n$ be the sequence
of concrete states visited along~$\pi$ in~$\TA$, that is, $B_1$ is the
concrete initial state, and for all~$2\leq i \leq n$, let~$B_i =
\postta^c_{\pi_{i-1}}(B_{i-1})$.  This~sequence can be computed using
DBMs.

The trace is \emph{realizable} if~$B_n \neq \emptyset$, in which case
the counterexample is confirmed. Otherwise it is \emph{spurious}.
We~show how to refine the abstraction to eliminate a spurious trace
$\pi$.

\smallskip

Let~$i_0$ be the maximal index such that~$B_{i_0} \neq \emptyset$. There are three possible reasons explaining why~$B_{i_0+1}$ is empty:
\begin{enumerate}
\item first, if the abstract successor~$A_{i_0+1}$ is unsatisfiable,
  that is, if it~contains contradictory predicates;
  in this case, $\sem{A_{i_0+1}} = \emptyset$, and the abstraction is
  refined by Lemma~\ref{lemma:empty-refine} to eliminate this case by
  strengthening $\reduce^k_\Dom$.
\item if there are predecessors of~$A_{i_0+1}$ inside~$A_{i_0}$ but none of them are in~$B_{i_0}$, i.e., $\preta_{\pi_{i_0}}^c(\sem{A_{i_0+1}}) \cap \sem{A_{i_0}} \neq \emptyset$;
in this case, we~refine the domain by separating these predecessors from the rest of~$A_{i_0}$ using $\interp(\preta_{\pi_{i_0}}^c(\sem{A_{i_0+1}}),B_{i_0-1})$, as in~\cite{CGJ+03}.
\item otherwise, there are no predecessors of~$A_{i_0+1}$ inside~$A_{i_0}$: we refine the abstraction according to the type of the transition
from step~$i_0$ to~$i_0+1$:
      \begin{enumerate}
      \item if $\pi_{i_0}=\up$: refine~$\Dom$ by~$\interp(\sem{A_{i_0}}{\uparrow}, \sem{A_{i_0+1}}{\downarrow})$.
      \item if~$\pi_{i_0}=r(x)$: refine~$\Dom$ by~$\interp(\Free_x(\sem{A_{i_0}}), \Free_x(\sem{A_{i_0+1}}))$.
      \end{enumerate}
\end{enumerate}
Note that the case~$\pi_{i_0}=\rempty$ is not possible since this
induces the identity function both in the abstract and concrete
systems.

Given abstraction~$\alpha_\Dom$ and spurious trace~$\pi$, let~$\refine(\alpha_\Dom,\pi)$
denote the refined abstraction~$\alpha_{\Dom'}$ obtained as described above.
%
%
%
%
The following two lemmas justify the two subcases of the third case
above.  They prove that the detected spurious transition disappears
after refinement.  The reset and up operations depend on the
abstraction, so we make this dependence explicit below by using
superscripts, as in $\Reset_x^\alpha$ and~$\Up^\alpha$, in order to
distinguish the operations before and after a refinement.
\begin{restatable}{lemma}{lemmarefineup}
\label{lemma:refine-up}
  Consider~$(A_1,A_2) \in \Up^\alpha$ with~$\timesucc{\sem{A_1}} \cap
  \sem{A_2} = \emptyset$.  Then $\timesucc{\sem{A_{1}}} \cap
  \timepred{\sem{A_{2}}} = \emptyset$.  Moreover, if~$\alpha'$ is
  obtained by refinement of~$\alpha$ by
  $\interp(\timesucc{\sem{A_{1}}}, \timepred{\sem{A_{2}}})$, then for
  all~$(A_1',A_2') \in \Up^{\alpha'}$, $\sem{A_1'} \subseteq
  \sem{A_1}$ implies~$\sem{A_2'} \cap \sem{A_2} = \emptyset$.
\end{restatable}

\begin{figure}[t]
  \centering
\begin{subfigure}[t]{0.48\textwidth}
  \centering
  \begin{tikzpicture}
    \begin{scope}[scale=.35]
      \fill[black!40] (0,0) -- (2,0) -- (2,1) -- (1,1)-- cycle;
      \draw[-latex',line width=.6pt] (0,0) -- (5.5,0);
      \draw[-latex',line width=.6pt] (0,0) -- (0,4.5);
      \fill[pattern=north west lines, pattern color=blue]
      (0,0) -- (4.5,4.5) -- (5.5,4.5) -- (5.5,1.5) -- (4,0) -- cycle;

      \node at (0,5) {$y$};
      \node at (6,0) {$x$};
      \node at (1.3,-1.5) (A1) {$A_1$};
      \node at (3.8,-1.5) (A2) {$A_2$};
      \draw[-latex'] (A1) -- ($(A1)+(0,2)$);
      \draw[-latex'] (A2) -- ($(A2)+(0,2)$);

      \begin{scope}[line width=0.1pt]
        \draw[-] (1,0) -- (1,4.5);
        \draw[-] (2,0) -- (2,4.5);
        \draw[-] (3,0) -- (3,4.5);
        \draw[-] (4,0) -- (4,4.5);
        \draw[-] (5,0) -- (5,4.5);
        \draw[-] (0,1) -- (5.5,1);
        \draw[-] (0,2) -- (5.5,2);
        \draw[-] (0,3) -- (5.5,3);
        \draw[-] (0,4) -- (5.5,4);
        \draw[-] (0,0) -- (4.5,4.5);
        \draw[-] (1,0) -- (5.5,4.5);
        \draw[-] (2,0) -- (5.5,3.5);
        \draw[-] (3,0) -- (5.5,2.5);
        \draw[-] (0,1) -- (3.5,4.5);
        \draw[-] (0,2) -- (2.5,4.5);
        \draw[-] (0,3) -- (1.5,4.5);
      \end{scope}
      \begin{scope}[line width=1.5pt,color=red]
        \draw[-] (0,1) -- (5.5,1);
        \draw[-] (2,0) -- (2,4.5);
        \draw[-] (3,0) -- (3,4.5);
        \draw[-] (0,0) -- (4.5,4.5);
        \draw[-] (4,0) -- (5.5,1.5);
      \end{scope}
    \end{scope}
  \end{tikzpicture}
  \begin{tikzpicture}
    \begin{scope}[scale=.35]
      \draw[-latex',line width=.6pt] (0,0) -- (5.5,0);
      \draw[-latex',line width=.6pt] (0,0) -- (0,4.5);
      \fill[pattern=north west lines, pattern color=blue]
      (0,0) -- (2,0) -- (5.5,3.5) -- (5.5,4.5) -- (4.5,4.5) -- cycle;
      \node at (0,5) {$y$};
      \node at (6,0) {$x$};
      \node at (1.3,-1.5) (A1) {$A_1$};
      \node at (3.8,-1.5) (A2) {$A_2$};
      \draw[-latex'] (A1) -- ($(A1)+(0,2)$);
      \draw[-latex'] (A2) -- ($(A2)+(0,2)$);
      \begin{scope}[line width=0.1pt]
        \draw[-] (1,0) -- (1,4.5);
        \draw[-] (2,0) -- (2,4.5);
        \draw[-] (3,0) -- (3,4.5);
        \draw[-] (4,0) -- (4,4.5);
        \draw[-] (5,0) -- (5,4.5);
        \draw[-] (0,1) -- (5.5,1);
        \draw[-] (0,2) -- (5.5,2);
        \draw[-] (0,3) -- (5.5,3);
        \draw[-] (0,4) -- (5.5,4);
        \draw[-] (0,0) -- (4.5,4.5);
        \draw[-] (1,0) -- (5.5,4.5);
        \draw[-] (2,0) -- (5.5,3.5);
        \draw[-] (3,0) -- (5.5,2.5);
        \draw[-] (0,1) -- (3.5,4.5);
        \draw[-] (0,2) -- (2.5,4.5);
        \draw[-] (0,3) -- (1.5,4.5);
      \end{scope}
      \begin{scope}[line width=1.5pt,color=red]
        \draw[-] (0,1) -- (5.5,1);
        \draw[-] (2,0) -- (2,4.5);
        \draw[-] (3,0) -- (3,4.5);
        \draw[-] (0,0) -- (4.5,4.5);
        \draw[-] (4,0) -- (5.5,1.5);
        \draw[-] (2,0) -- (5.5,3.5);
      \end{scope}
    \end{scope}
  \end{tikzpicture}
  \caption{Refinement for the time successors operation.
    The interpolant that separates $\sem{A_1}{\uparrow}$ from~$\sem{A_2}{\downarrow}$
    contains the constraint~$x=y+2$. When this is added to the abstract domain,
    the set $A_2'$ (which is $A_2$ in the new abstraction)
    is no longer reachable by the time successors operation.
}
  \label{fig:refinement-up}
\end{subfigure}
\hfil
\begin{subfigure}[t]{0.45\textwidth}
  \begin{tikzpicture}
    \begin{scope}[scale=.35,shift={(0,0)}]
      \fill[black!40] (0,1) -- (2,3) -- (0,3) -- cycle;
     \fill[pattern=north west lines, pattern color=blue] (0,0) -- (4,0) -- (4,1) -- (0,1) -- cycle;
     \draw[-latex',line width=.6pt] (0,0) -- (5.5,0);
      \draw[-latex',line width=.6pt] (0,0) -- (0,4.5);
      \node at (0,5) {$y$};
      \node at (6,0) {$x$};
      \node at (-1,2.3) (A1) {$A_1$};
      \node at (3.3,-1.5) (A2) {$A_2$};
      \draw[-latex'] (A1) -- ($(A1)+(1.6,0)$);
      \draw[-latex'] (A2) -- ($(A2)+(0,2)$);
      \begin{scope}[line width=0.1pt]
        \draw[-] (1,0) -- (1,4.5);
        \draw[-] (2,0) -- (2,4.5);
        \draw[-] (3,0) -- (3,4.5);
        \draw[-] (4,0) -- (4,4.5);
        \draw[-] (5,0) -- (5,4.5);
        \draw[-] (0,1) -- (5.5,1);
        \draw[-] (0,2) -- (5.5,2);
        \draw[-] (0,3) -- (5.5,3);
        \draw[-] (0,4) -- (5.5,4);
        \draw[-] (0,0) -- (4.5,4.5);
        \draw[-] (1,0) -- (5.5,4.5);
        \draw[-] (2,0) -- (5.5,3.5);
        \draw[-] (3,0) -- (5.5,2.5);
        \draw[-] (0,1) -- (3.5,4.5);
        \draw[-] (0,2) -- (2.5,4.5);
        \draw[-] (0,3) -- (1.5,4.5);
      \end{scope}
      \begin{scope}[line width=1.5pt,color=red]
        \draw[-] (0,1) -- (3.5,4.5);
        \draw[-] (0,3) -- (5.5,3);
        \draw[-] (4,0) -- (4,4.5);
        \draw[-] (0,1) -- (5.5,1);
        \draw[-] (2,0) -- (5.5,3.5);
      \end{scope}
    \end{scope}
    \begin{scope}[scale=.35,shift={(8,0)}]
      \fill[black!40] (0,1) -- (2,3) -- (0,3) -- cycle;
     \fill[pattern=north west lines, pattern color=blue] (0,0) -- (2,0) -- (2,1) -- (0,1) -- cycle;
     \draw[-latex',line width=.6pt] (0,0) -- (5.5,0);
      \draw[-latex',line width=.6pt] (0,0) -- (0,4.5);
      \node at (0,5) {$y$};
      \node at (6,0) {$x$};

      \node at (-1.2,2.3) (A1) {$A_1'$};
      \node at (3.3,-1.5) (A2) {$A_2'$};
      \draw[-latex'] (A1) -- ($(A1)+(1.5,0)$);
      \draw[-latex'] (A2) -- ($(A2)+(0,2)$);

      \begin{scope}[line width=0.1pt]
        \draw[-] (1,0) -- (1,4.5);
        \draw[-] (2,0) -- (2,4.5);
        \draw[-] (3,0) -- (3,4.5);
        \draw[-] (4,0) -- (4,4.5);
        \draw[-] (5,0) -- (5,4.5);
        \draw[-] (0,1) -- (5.5,1);
        \draw[-] (0,2) -- (5.5,2);
        \draw[-] (0,3) -- (5.5,3);
        \draw[-] (0,4) -- (5.5,4);
        \draw[-] (0,0) -- (4.5,4.5);
        \draw[-] (1,0) -- (5.5,4.5);
        \draw[-] (2,0) -- (5.5,3.5);
        \draw[-] (3,0) -- (5.5,2.5);
        \draw[-] (0,1) -- (3.5,4.5);
        \draw[-] (0,2) -- (2.5,4.5);
        \draw[-] (0,3) -- (1.5,4.5);
      \end{scope}
      \begin{scope}[line width=1.5pt,color=red]
        \draw[-] (0,1) -- (3.5,4.5);
        \draw[-] (0,3) -- (5.5,3);
        \draw[-] (4,0) -- (4,4.5);
        \draw[-] (0,1) -- (5.5,1);
        \draw[-] (2,0) -- (5.5,3.5);
        \draw[-] (2,0) -- (2,4.5);
      \end{scope}
    \end{scope}

  \end{tikzpicture}
  \caption{Refinement for the reset operation.
    The interpolant that separates $\Free_y(A_1)$ from~$\Free_y(A_2)$
    contains the constraint~$x< 2$. When this is added to the abstract domain,
    the set $A_2'$ (which is $A_2$ in the new abstraction)
    is no longer reachable by the reset operation.
}
  \label{fig:reset-refinement}
\end{subfigure}
\end{figure}

\begin{proof}
  Assume there is~$v \in \timesucc{\sem{A_{1}}} \cap
  \timepred{\sem{A_{2}}}$. There exists~$d_1,d_2\geq 0$ and~$v_1 \in
  A_1$ such that $v = v_1 + d_1$ and~$v+d_2\in A_2$, which
  means~$v_1+d_1+d_2 \in A_2$, thus $\sem{A_1}{\uparrow} \cap
  \sem{A_2}\neq \emptyset$.

  Let~$Z = \sem{\interp(\timesucc{\sem{A_{1}}}, \timepred{\sem{A_{2}}})}$.
  By definition, $\timesucc{\sem{A_1}} \subseteq Z$, and~$Z \cap \timepred{\sem{A_2}} = \emptyset$.
  We have~$\timesucc{Z} = Z$; in~fact, $Z$~cannot have upper bounds since
  $\timesucc{\sem{A_1}}$ does not have~any.
  It~follows that~$\timesucc{Z} \cap \timepred{\sem{A_2}} = \emptyset$.

  Now, consider~$(A_1',A_2') \in \Up^{\alpha'}$ with~$\sem{A_1'}
  \subseteq \sem{A_1}$.
  Let us show that~$\sem{A_2'} \subseteq Z$.  Notice that all
  constraints of~$Z$ must be satisfied by~$A_2'$ due to the inclusion
  $\sem{A_1'} \subseteq Z$: there are no upper bounds in~$Z$, all its
  lower bounds are satisfied by~$A_1$, thus also by~$A_1'$, and are
  preserved by definition by~$\Up$, and all diagonal constraints
  of~$Z$ hold in~$A_1$, thus also in~$A_1'$, and are preserved as well
  in~$A_2'$.  This shows the required inclusion.  It~follows
  that~$\sem{A_2'} \cap \sem{A_2} = \emptyset$.
\end{proof}

\begin{restatable}{lemma}{lemmarefinereset}
\label{lemma:refine-reset}
  Consider~$x \in \Clocks$, and $(A_1,A_2) \in \Reset_x^\alpha$ such
  that~$\sem{A_1}[x\leftarrow 0] \cap \sem{A_2} = \emptyset$.  Then
  $\Free_x(\sem{A_{1}}) \cap \Free_x(\sem{A_{2}}) = \emptyset$.
  Moreover, if $\alpha'$ is obtained by refinement of~$\alpha$ by
  $\interp(\Free_x(\sem{A_{1}}), \Free_x(\sem{A_{2}}))$, then for
  all~$(A_1',A_2') \in \Reset_x^{\alpha'}$ with~$\sem{A_1'} \subseteq
  \sem{A_1}$, we have $\sem{A_2'} \cap \sem{A_2} = \emptyset$.
\end{restatable}

\begin{proof}
  Let $v \in \Free_x(\sem{A_{1}}) \cap \Free_x(\sem{A_{2}})$. Then
  there exist~$v_1 \in \sem{A_1}$, $v_2 \in \sem{A_2}$
  and~$v_0$ such
  that~$v_0 = v[x\leftarrow 0] = v_1[x\leftarrow 0] = v_2[x\leftarrow 0]$.
  But~$\sem{A_2}$ is
  closed by resetting~$x$, that is, $\sem{A_2}[x:=0] \subseteq
  \sem{A_2}$.  This follows from Lemma~\ref{lemma:op-reset} applied
  to~$A_2$, by~observing that~$A_2$ is unchanged by the reset
  operation.  So~$v_0 \in \sem{A_2}$. But then~$\sem{A_1}[x:=0] \cap
  \sem{A_2} \neq \emptyset$ as witnessed by~$v_1[x:=0] = v_0$.
  
  Let $Z = \sem{\interp(\Free_x(\sem{A_1}),\Free_x(\sem{A_2}))}$.  For
  all~$A_1'$ satisfying~$\sem{A_1'} \subseteq \sem{A_1}$, we
  have 
  $\sem{A_1'}[x \leftarrow 0] \subseteq \sem{A_1}[x\leftarrow 0] \subseteq \Free_x(\sem{A_1}) \subseteq
  Z$. So~$Z \cap A_2 =\emptyset$ means that ${\sem{A_1'}[x\leftarrow 0] \cap
    A_2} = \emptyset$.
\end{proof}

\section{Experiments}
We implemented both algorithms. The symbolic version was implemented
in OCaml using the \texttt{CUDD}
library\footnote{\url{http://vlsi.colorado.edu/~fabio/}};
the explicit version was implemented in C++ within an existing
model checker using Uppaal DBM library.  Both prototypes take as input
networks of timed
automata with invariants, discrete variables, urgent and committed
locations. The~presented algorithms are adapted to these features
without difficulty.

We evaluated our algorithms on three classes of benchmarks we believe
are significant.  We~compare the performance of the algorithm with
that of Uppaal~\cite{BDL+06} which is based on zones, as well as
the BDD-based model checker engine of PAT~\cite{NguyenSLDL12}.  We were unable to
compare with RED~\cite{farn-ftnds01} which is not maintained anymore
and not open source, and with which we failed to obtain correct
results.  The tool used in~\cite{EFGP-rtss2010} was not available
either. We thus only provide a comparison here with two
well-maintained tools.

Two of our benchmarks are variants of schedulability-analysis problems
where task execution times depend on the internal states of executed processes,
so that an analysis of the state space is necessary to obtain a precise answer.

\noindent\textbf{Monoprocess Scheduling Analysis.}
In this variant, a single process sequentially executes tasks on a
single machine, and the execution time of each cycle depends on the
state of the process.  The goal is to determine a bound on the maximum
execution time of a single cycle.  This depends on the semantics of
the process since the bound depends on the reachable states.

More precisely, we built a set of benchmarks where the processes are
defined by synchronous circuit models taken from the Synthesis
Competition (\url{http://www.syntcomp.org}).  We assume that each
latch of the circuit is associated with a resource, and changing the
state of the resource takes some amount of time.  So a subset of the
latches have clocks associated with them, which measure the time
elapsed since the latest value change (latest moment when the value
changed from~$0$ to~$1$, or from~$1$ to~$0$). We provide two time
positive bounds~$\ell_0$ and $\ell_1$ for each latch, which determine
the execution time as follows: if~the value of latch~$\ell$ changes
from~$0$ to~$1$ (resp.~from $1$ to~$0$), then the execution time of
the present cycle cannot be less than~$\ell_1$ (resp.~$\ell_0$).
The~execution time of the step is then the minimum that satisfies
these constraints.

\noindent\textbf{Multi-process Stateful Scheduling Analysis.}
In this variant, three processes are scheduled on two machines with a
round-robin policy.  Processes schedule tasks one after the other
without any delay. As~in the previous benchmarks, a~process executing a
task (on~any machine) corresponds to a step of the synchronous circuit
model. Each~task is described by a tuple $(C_1,C_2,D)$ which defines the
minimum and maximum execution times, and the relative deadline. When a task
finishes, the~next task arrives immediately. The~values in the tuple
depend on the state of the process. The~goal is to check the absence
of any deadline miss.
Processes are also instantiated with AIG circuits from \url{http://www.syntcomp.org}.

\noindent\textbf{Asynchronous Computation.}
We consider an asynchronous network of ``threshold gates'', defined as
follows: each gate is characterized by a tuple $(n,\theta,[l,u])$
where~$n$ is the number of inputs, $0\leq \theta \leq n$ is the
threshold, and~$l\leq u$ are lower and upper bounds on activation
time.  Each gate has an output which is initially undefined. The~gate
becomes active during the time period $[l,u]$. During this time, if all
inputs are defined, and if at least~$\theta$ of the inputs have
value~$1$, then it~sets its output to~$1$. At~the end of the time period,
it~becomes deactivated and the output becomes undefined again, until
the next period, which starts $l$ time units after the
deactivation.  The~goal is to check whether the given gate can
output~$1$ within a given time bound~$T$.

\begin{figure}[ht]
  \centering
  \includegraphics[scale=0.4]{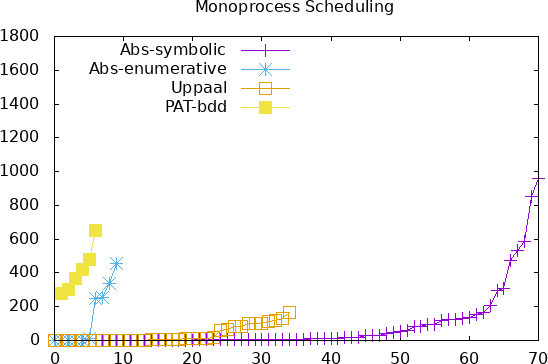}
  \includegraphics[scale=0.4]{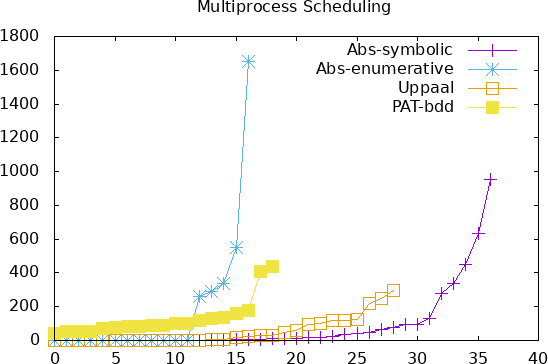}

\bigskip
  \includegraphics[scale=0.4]{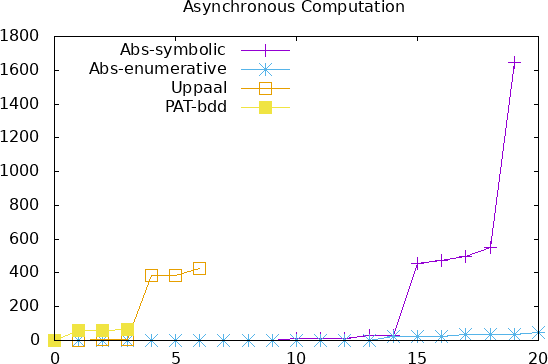}
  \includegraphics[scale=0.4]{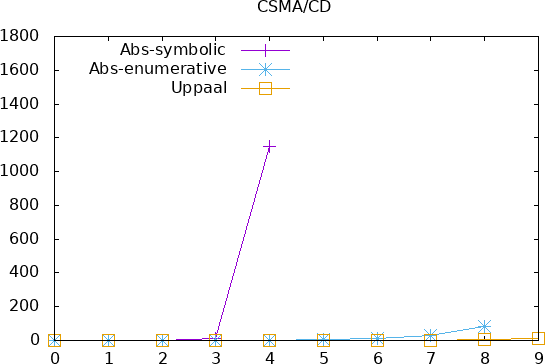}
  \caption{Comparison of our enumerative and symbolic algorithms
    (resp. Abs-enumerative and Abs-symbolic) with Uppaal and~PAT.
    Each figure is a cactus plot for the set of benchmarks: a~point
    (X,Y) means X benchmarks were solved within time bound
    Y.}\label{fig-results}
\end{figure}

\noindent\textbf{Results.}
%
Figure~\ref{fig-results} displays the results of our experiments.
All algorithms were given 8GB of memory and a timeout of 30 minutes,
and the experiments were run on laptop with an Intel i7@3.2Ghz
processor running Linux.  The symbolic algorithm performs best among
all on the monoprocess and multiprocess scheduling
benchmarks. Uppaal is the second best, but does not solve as many
benchmarks as our algorithm. Our enumerative algorithm quickly fails
on these benchmarks, often running out of memory.  On asynchronous
computation benchmarks, our enumerative algorithm performs
remarkably well, beating all other algorithms.
We ran our tools on the CSMA/CD benchmarks (with 3 to
12 processes); Uppaal performs the best but our enumerative algorithm
is slightly behind. The~symbolic algorithm does not scale, while PAT fails to terminate in all cases.

The tool used for the symbolic algorithm is open source and can be found
at \url{https://github.com/osankur/symrob} along with all the benchmarks.


\section{Conclusion and Future Work}

There are several ways to improve the
algorithm. Since the choice of interpolants determines the abstraction
function and the number of refinements, we~assumed that taking the
minimal interpolant should be preferable as it should keep the
abstractions as coarse as possible. But it might be better to predict
which interpolant is the most adapted for the rest of the computation
in order to limit future refinements. The~number of refinement also
depends on the search order, and although it has already been
studied in~\cite{DBLP:conf/formats/HerbreteauT15}, it~could be
interesting to study it in this case.
Generally speaking, it is  worth noting that we currently cannot
predict which (variant~of) our algorithms is better suited for which
model.

Several extensions of our algorithms could  be developed,
\textit{e.g.}
combining our algorithms
with other methods based on finer abstractions as
in~\cite{HSW-cav13},
integrating predicate abstraction on discrete variables, or
developing SAT-based versions of our algorithms.

\bibliographystyle{abbrv}
\bibliography{bibexport}

\clearpage
\appendix

\section{Details on Benchmarks}
\label{app:bench}
Monoprocess-scheduling benchmarks we considered are listed below.  In
each case, the \texttt{.aag} file is the circuit defining the process
(the~models with identical file names is available at
\url{http://www.syntcomp.org}).  Only a subset of the latches are
``clocked'', that have time constraints. This can be seen in the file
name: for~instance, in~\texttt{amba3b5y.aag\_4L\_200.xml}, only the
first four latches are clocked.  The~number of clocks is then~five
(an~additional one is used to test the elapsed global time).  The~last
number is the time bound to be tested.  The~complete list of
benchmarks are given in Table~\ref{table:mono}.

\def\arraystretch{1.1}
\begin{longtable}{|>{\quad\ttfamily}l<{\quad}|}
\hline\endhead
\hline\endfirsthead
\hline\endfoot
\endlastfoot
amba3b5y.aag\_10L\_300.xml \\\hline
amba3b5y.aag\_4L\_200.xml \\\hline
amba3b5y.aag\_4L\_290.xml \\\hline
amba3b5y.aag\_4L\_300.xml \\\hline
amba3b5y.aag\_5L\_290.xml \\\hline
amba3b5y.aag\_5L\_300.xml \\\hline
amba3b5y.aag\_6L\_290.xml \\\hline
amba3b5y.aag\_6L\_300.xml \\\hline
amba3b5y.aag\_7L\_290.xml \\\hline
amba3b5y.aag\_7L\_300.xml \\\hline
amba3b5y.aag\_8L\_300.xml \\\hline
amba3b5y.aag\_9L\_300.xml \\\hline
amba4c7y.aag\_10L\_300.xml \\\hline
amba4c7y.aag\_4L\_200.xml \\\hline
amba4c7y.aag\_4L\_300.xml \\\hline
amba4c7y.aag\_5L\_200.xml \\\hline
amba4c7y.aag\_5L\_300.xml \\\hline
amba4c7y.aag\_6L\_300.xml \\\hline
amba4c7y.aag\_7L\_300.xml \\\hline
amba4c7y.aag\_8L\_300.xml \\\hline
amba4c7y.aag\_9L\_300.xml \\\hline
bs16y.aag\_4L\_100.xml \\\hline
bs16y.aag\_4L\_150.xml \\\hline
bs16y.aag\_4L\_200.xml \\\hline
cnt5y.aag\_4L\_200.xml \\\hline
cnt5y.aag\_4L\_300.xml \\\hline
factory\_assembly\_3x3\_1\_1errors.aag\_10L\_500.xml \\\hline
factory\_assembly\_3x3\_1\_1errors.aag\_4L\_200.xml \\\hline
factory\_assembly\_3x3\_1\_1errors.aag\_4L\_300.xml \\\hline
factory\_assembly\_3x3\_1\_1errors.aag\_5L\_300.xml \\\hline
factory\_assembly\_3x3\_1\_1errors.aag\_6L\_300.xml \\\hline
factory\_assembly\_3x3\_1\_1errors.aag\_6L\_400.xml \\\hline
factory\_assembly\_3x3\_1\_1errors.aag\_6L\_500.xml \\\hline
factory\_assembly\_3x3\_1\_1errors.aag\_7L\_500.xml \\\hline
factory\_assembly\_3x3\_1\_1errors.aag\_8L\_500.xml \\\hline
factory\_assembly\_3x3\_1\_1errors.aag\_9L\_500.xml \\\hline
genbuf2b3unrealy.aag\_10L\_400.xml \\\hline
genbuf2b3unrealy.aag\_4L\_300.xml \\\hline
genbuf2b3unrealy.aag\_5L\_250.xml \\\hline
genbuf2b3unrealy.aag\_5L\_300.xml \\\hline
genbuf2b3unrealy.aag\_6L\_300.xml \\\hline
genbuf2b3unrealy.aag\_7L\_300.xml \\\hline
genbuf2b3unrealy.aag\_7L\_400.xml \\\hline
genbuf2b3unrealy.aag\_8L\_400.xml \\\hline
genbuf2b3unrealy.aag\_9L\_400.xml \\\hline
genbuf5f5n.aag\_10L\_300.xml \\\hline
genbuf5f5n.aag\_5L\_290.xml \\\hline
genbuf5f5n.aag\_5L\_300.xml \\\hline
genbuf5f5n.aag\_6L\_290.xml \\\hline
genbuf5f5n.aag\_6L\_300.xml \\\hline
genbuf5f5n.aag\_7L\_290.xml \\\hline
genbuf5f5n.aag\_7L\_300.xml \\\hline
genbuf5f5n.aag\_8L\_300.xml \\\hline
genbuf5f5n.aag\_9L\_300.xml \\\hline
moving\_obstacle\_8x8\_1glitches.aag\_10L\_300.xml \\\hline
moving\_obstacle\_8x8\_1glitches.aag\_4L\_150.xml \\\hline
moving\_obstacle\_8x8\_1glitches.aag\_4L\_300.xml \\\hline
moving\_obstacle\_8x8\_1glitches.aag\_5L\_150.xml \\\hline
moving\_obstacle\_8x8\_1glitches.aag\_5L\_300.xml \\\hline
moving\_obstacle\_8x8\_1glitches.aag\_6L\_150.xml \\\hline
moving\_obstacle\_8x8\_1glitches.aag\_6L\_300.xml \\\hline
moving\_obstacle\_8x8\_1glitches.aag\_7L\_150.xml \\\hline
moving\_obstacle\_8x8\_1glitches.aag\_7L\_300.xml \\\hline
moving\_obstacle\_8x8\_1glitches.aag\_8L\_300.xml \\\hline
moving\_obstacle\_8x8\_1glitches.aag\_9L\_300.xml \\\hline
\caption{Monoprocess benchmarks}
\label{table:mono}
\end{longtable}

\medskip
For the multiprocess-scheduling benchmarks, we generated
instances using the data shown in Table~\ref{table:multi}. 
All models have three clocks, one per process.
The~first three entries
show the circuits (from \url{http://www.syntcomp.org}) used to define
the processes that are being executed. The last number is the number
of the scenario,
which determines the execution times of arriving tasks according
to the value of a selected latch:
\begin{itemize}
\item in scenario 0, the two tuples~$(C_1,C_2,D)$ of execution time interval
  and relative deadlines are: $(500,1000,3000), (400,800,3000)$.
\item in scenario 1: $(500,1000, 1500), (400, 800, 1600)$.
\item in scenario 2: $(1000, 1000, 10000), (20000, 20000, 200000)$.
\end{itemize}

\begin{longtable}{|>{\ }c<{\ }|>{\ \ttfamily}c<{\ }|>{\ \ttfamily}c<{\ }|>{\ \ttfamily}c<{\ }|>{\ }c<{\ }|}
  \hline
  Model & Process1 & Process2 & Process3 & Scenario \\  \hline\hline\endhead
  \hline
  Model & Process1 & Process2 & Process3 & Scenario \\  \hline\hline\endfirsthead
\endfoot
\endlastfoot
0 & amba3b5y.aag & add2y.aag & add2y.aag & 0\\ \hline
1 & amba3b5y.aag & add2y.aag & add2y.aag & 1\\ \hline
2 & amba3b5y.aag & add2y.aag & add2y.aag & 2\\ \hline
3 & cnt4y.aag & cnt3y.aag & cnt3y.aag & 0\\ \hline
4 & cnt4y.aag & cnt3y.aag & cnt3y.aag & 1\\ \hline
5 & cnt4y.aag & cnt3y.aag & cnt3y.aag & 2\\ \hline
6 & cnt4y.aag & cnt4y.aag & cnt3y.aag & 0\\ \hline
7 & cnt4y.aag & cnt4y.aag & cnt3y.aag & 1\\ \hline
8 & cnt4y.aag & cnt4y.aag & cnt3y.aag & 2\\ \hline
9 & cnt4y.aag & cnt4y.aag & cnt4y.aag & 0\\ \hline
10 & cnt4y.aag & cnt4y.aag & cnt4y.aag & 1\\ \hline
11 & cnt4y.aag & cnt4y.aag & cnt4y.aag & 2\\ \hline
12 & cnt5y.aag & cnt4y.aag & cnt3y.aag & 0\\ \hline
13 & amba3b5y.aag & cnt3y.aag & cnt3y.aag & 2\\ \hline
14 & cnt5y.aag & cnt3y.aag & cnt3y.aag & 0\\ \hline
15 & cnt5y.aag & cnt3y.aag & cnt3y.aag & 1\\ \hline
16 & cnt5y.aag & cnt3y.aag & cnt3y.aag & 2\\ \hline
17 & cnt3y.aag & cnt3y.aag & cnt3y.aag & 0\\ \hline
18 & cnt3y.aag & cnt3y.aag & cnt3y.aag & 1\\ \hline
19 & cnt3y.aag & cnt3y.aag & cnt3y.aag & 2\\ \hline
20 & amba3b5y.aag & add2y.aag & add2y.aag & 1\\ \hline
21 & amba3b5y.aag & add2y.aag & add2y.aag & 0\\ \hline
22 & amba3b5y.aag & add2y.aag & cnt3y.aag & 1\\ \hline
23 & amba3b5y.aag & add2y.aag & cnt3y.aag & 0\\ \hline
24 & bs8y.aag & add2y.aag & add2y.aag & 0\\ \hline
25 & bs8y.aag & add2y.aag & add2y.aag & 1\\ \hline
26 & bs8y.aag & add2y.aag & add2y.aag & 2\\ \hline
27 & bs8y.aag & bs8y.aag & add2y.aag & 0\\ \hline
28 & bs8y.aag & bs8y.aag & add2y.aag & 1\\ \hline
29 & bs8y.aag & bs8y.aag & add2y.aag & 2\\ \hline
30 & bs8y.aag & bs8y.aag & bs8y.aag & 0\\ \hline
31 & bs8y.aag & bs8y.aag & bs8y.aag & 1\\ \hline
32 & bs8y.aag & bs8y.aag & bs8y.aag & 2\\ \hline
33 & mv4y.aag & mv4y.aag & add2y.aag & 0\\ \hline
34 & mv4y.aag & mv4y.aag & add2y.aag & 1\\ \hline
35 & mv4y.aag & mv4y.aag & add2y.aag & 2\\ \hline
36 & mv4y.aag & mv4y.aag & mv4y.aag & 2\\ \hline
37 & stay2y.aag & stay2y.aag & mv4y.aag & 2\\ \hline
38 & stay4y.aag & add2y.aag & add2y.aag & 2\\ \hline
39 & stay4y.aag & cnt4y.aag & add2y.aag & 2\\ \hline
40 & stay4y.aag & stay2y.aag & mv4y.aag & 2\\ \hline
41 & stay4y.aag & stay2y.aag & stay2y.aag & 2\\ \hline
\caption{Multiprocess benchmarks}
\label{table:multi}
\end{longtable}

\medskip

Information on asynchronous-computation benchmarks is listed in
Table~\ref{table:async}.  The number of clocks in each model is equal
to the number of non-input gates.

\begin{longtable}{|>{\ \ttfamily}c<{\ }|>{\ }c<{\ }|>{\ }c<{\ }|>{\ }c<{\ }|}
\hline File name & Number of gates & Number of inputs & Time bound \\\hline\hline\endhead
\hline File name & Number of gates & Number of inputs & Time bound \\\hline\hline\endfirsthead
\endfoot\endlastfoot
a0 & 8 & 4 & 50\\ \hline
a1 & 8 & 4 & 150\\ \hline
a2 & 9 & 4 & 50\\ \hline
a3 & 9 & 4 & 150\\ \hline
a4 & 9 & 4 & 400\\ \hline
a5 & 16 & 8 & 50\\ \hline
a6 & 16 & 8 & 150\\ \hline
a7 & 19 & 14 & 150\\ \hline
a8 & 19 & 14 & 300\\ \hline
a9 & 20 & 14 & 300\\ \hline
a10 & 20 & 14 & 300\\ \hline
b0 & 9 & 4 & 1000\\ \hline
b1 & 10 & 4 & 1000\\ \hline
b2 & 9 & 4 & 1000\\ \hline
b3 & 16 & 8 & 1000\\ \hline
b4 & 40 & 35 & 1000\\ \hline
b5 & 20 & 14 & 1000\\ \hline
b6 & 20 & 15 & 1000\\ \hline
b7 & 9 & 4 & 1000\\ \hline
b8 & 10 & 3 & 1000\\ \hline
b9 & 19 & 14 & 1000\\ \hline
b10 & 16 & 8 & 1000\\ \hline
\caption{Asynchronous Computation Benchmarks}
\label{table:async}
\end{longtable}

\end{document}